\documentclass[aps,pra,floatfix,showpacs,twocolumn,tightenlines,superscriptaddress,nofootinbib,notitlepage]{revtex4-1}%
\usepackage{amsmath,color}
\usepackage{graphicx}
\usepackage[bookmarks=false]{hyperref}
\usepackage{bm}
\usepackage{amsfonts}
\usepackage{amsthm}
\usepackage{amssymb}
\usepackage{enumitem}
\hypersetup{
    pdftoolbar=true,        
    pdfmenubar=true,        
    pdffitwindow=true,      
    pdfstartview={},        
    pdftitle={Maximally entangled mixed states for qubit-qutrit systems}                    
}

\newcommand\T{\rule{0pt}{3.0ex}}       
\DeclareMathOperator{\tr}{tr}

\def\ket#1{\mathinner{|{#1}\rangle}}
\def\coloneq{\mathrel{\mathop:}=}

\newtheorem{theorem}{Theorem}

\newtheorem{lemma}{Lemma}[section]
\newtheorem{proposition}{Proposition}[section]

\begin{document}
\title{Maximally entangled mixed states for qubit-qutrit systems}
\author{Paulo E. M. F. Mendon\c{c}a}\email{pmendonca@gmail.com}
\affiliation{Divis\~ao de Ensino, Academia da For\c{c}a A\'erea, 13.643-000 Pirassununga, SP, Brazil}
\author{Marcelo A. Marchiolli}\email{marcelo march@bol.com.br}
\affiliation{Avenida General Os\'orio 414, centro, 14.870-100 Jaboticabal, SP, Brazil}
\author{Samuel R. Hedemann}\email{samuel.hedemann@gmail.com}
\affiliation{P.O. Box 72, Freeland, MD 21053, USA}
\date{\today}
\begin{abstract}
We consider the problems of maximizing the entanglement negativity of X-form qubit-qutrit density matrices with (i) a fixed spectrum and (ii) a fixed purity. In the first case, the problem is solved in full generality whereas, in the latter, partial solutions are obtained by imposing extra spectral constraints such as rank-deficiency and degeneracy, which enable a semidefinite programming treatment for the optimization problem at hand. Despite the technically-motivated assumptions, we provide strong numerical evidence that three-fold degenerate X states of purity $P$ reach the highest entanglement negativity accessible to arbitrary qubit-qutrit density matrices of the same purity, hence characterizing a sparse family of likely qubit-qutrit maximally entangled mixed states.
\end{abstract}

\maketitle

\section{Introduction}

Pure bipartite entanglement has long been a well-characterized quantum resource. There is, in essence, a unique way to quantify the amount of entanglement in a pure bipartite quantum state --- the von Neumann entropy of either of its reduced density matrices~\cite{96Bennett3824,97Popescu3319,02Donald4252}. Accordingly, such a state is said to be \emph{maximally entangled} if it maximizes this ``entropy of entanglement,'' a condition met by states for which all nonzero Schmidt coefficients are equal (e.g., Bell states of two qubits~\cite{00Nielsen}).

In sharp contrast, mixed bipartite entanglement is quite a formidable subject~\cite{09Horodecki865}. Although a number of well-defined and physically motivated mixed entanglement measures have been identified~\cite{01Horodecki3}, they are usually hard to compute and do not always agree on which one of two given entangled states is more entangled~\cite{99Eisert145,00Virmani31}. Moreover, even the seemingly simple problem of deciding whether a bipartite mixed state is entangled or not has been shown to be NP-hard~\cite{03Gurvits10,07Ioannou335}. Nevertheless, mixed entanglement detection in two-qubit ($2\times 2$) and qubit-qutrit ($2\times 3$) systems can be carried out with relative ease thanks to a remarkable result due to Peres~\cite{96Peres1413} and Horodecki~\cite{96Horodecki1}: the occurrence of a negative eigenvalue in the partially transposed density matrix of a $2\times 2$ or $2\times 3$ system is a necessary and \emph{sufficient} condition for mixed-state entanglement. 

Building on this criterion, Vidal and Werner defined the \emph{entanglement negativity} (or simply, \emph{negativity}), as a measure corresponding to a certain affine function of the sum of the magnitudes of the negative eigenvalues of either partially transposed density matrix, establishing it as an easy-to-compute entanglement measure \cite{02Vidal032314}.
Indeed, negativity obeys a number of interesting properties such as being invariant under local unitaries, monotonic under local operations and classical communications, and, in the case of $2\times 2$ and $2\times 3$ systems, zero if and only if no entanglement is present~\cite{96Peres1413,96Horodecki1}.

Resorting to negativity and concurrence (another computable entanglement measure for $2\times 2$ systems~\cite{98Wootters2245,01Wootters27}), Ishizaka and Hiroshima~\cite{00Ishizaka22310} introduced the concept of \emph{maximally entangled mixed states} (MEMS) to describe those states that cannot have their entanglement content increased by means of a global unitary transformation, or equivalently, states that --- for a fixed spectrum --- maximize a given entanglement measure. In Refs.~\cite{00Ishizaka22310,01Verstraete12316}, explicit forms for $2\times 2$ MEMS with fixed spectrum and maximal concurrence, negativity, and relative entropy of entanglement~\cite{97Vedral4452} were obtained. Subsequently, the notion of MEMS with a fixed mixedness (as measured by purity, linear entropy, or von Neumann entropy) was investigated and, once again, explicit forms of such $2\times 2$ MEMS were constructed~\cite{01Munro30302,03Wei22110}. 

Despite these early findings on $2\times 2$ systems, the problem of identifying $2\times 3$ MEMS has been, so far, mostly overlooked by the quantum information community. To the best of our knowledge, the only attempt to provide an explicit construction of such optimal states was made in Ref.~\cite{13Hedemann}, where numerical experimentation and educated guesses (inspired by two-qubit results) paved the way to a family of candidate states. Yet, a rigorous determination of actual $2\times 3$ MEMS is still a wide-open problem. This paper aims to contribute important first steps towards settling this issue.

Throughout most of this work, we constrain our search for $2\times 3$ MEMS to the subset of X states~\cite{07Yu459} --- a subset of density matrices that, written in the computational basis, cannot display nonzero entries outside of the main- and anti-main-diagonals. That is for two reasons: First, from a technical viewpoint, the sparsity of X states allows the relevant maximizations of negativity to be carried out analytically in many circumstances. Second, from the aforementioned results on two-qubit systems~\cite{00Ishizaka22310,01Verstraete12316,01Munro30302,03Wei22110} we learned that $2\times 2$ MEMS can always be written in X form (in fact, any $2\times 2$ state can be \emph{unitarily transformed} into an X state of the same entanglement~\cite{13Hedemann,14Mendonca79}), and so the hypothesis that $2\times 3$ MEMS will \emph{also} fall within this subset feels like a reasonable place to start. Notwithstanding, since this is ultimately an unconfirmed hypothesis, we name X-MEMS (as opposed to simply MEMS) the optimal density matrices resulting from optimization problems constrained to output X states. Whether or not X-MEMS are actual MEMS for $2\times 3$ systems remains an open problem, but strong numerical evidence supporting this conjecture is provided.

A combined summary of our main results and the structure of the paper is as follows. In Sec.~\ref{sec:preliminar} we derive a negativity formula for $2\times 3$ X states and provide a brief review of certain elements of semidefinite programming theory to be employed throughout. Sec.~\ref{sec:XMEMSspectrum} is devoted to the determination of the maximal negativity of $2\times 3$ X states constrained to have a given set of eigenvalues, culminating with an analytical construction of $2\times 3$ X-MEMS with respect to (wrt) spectrum. In Sec.~\ref{sec:spconstrainedXMEMS}, we move on to consider the characterization of $2\times 3$ X states that maximize negativity for a fixed value of purity.  The main results of this section are analytical constructions of $2\times 3$ X-MEMS wrt purity under three different circumstances: constrained to be of rank-2, rank-3, and without rank constraints but instead requiring the smallest eigenvalue to be three-fold degenerate. Grounding on these results, we show in Sec.~\ref{sec:universality} that, unlike it occurs in $2\times 2$ systems, not every $2\times 3$ density matrix can be \emph{unitarily} transformed into a $2\times 3$ X state of same negativity. Nevertheless, in Sec.~\ref{sec:MEMS}, we present the results of a numerical algorithm that supports the conjecture that $2\times 3$ MEMS wrt purity can always be made three-fold degenerate X states. We come to our conclusions in Sec.~\ref{sec:conclusion}.

\section{Preliminaries}\label{sec:preliminar}

This section aims to provide some (nonstandard) background material to be used throughout the paper. In Sec.~\ref{sec:nexstates}, the concept of an X state is formally introduced and a negativity formula for $2\times 3$ X states --- particularly suitable for these optimizations to be carried out --- is derived. In addition, Sec.~\ref{sec:SDP} gives a quick insight on certain aspects of the duality theory for semidefinite programs. For a more thorough review on semidefinite programming, we refer the reader to Refs.~\cite{96Vandenberghe49,04Boyd,02Audenaert030302}.

\subsection{Negativity of entanglement of $\mathbf{2\times 3}$ X states}\label{sec:nexstates}

An X state is any quantum state whose density matrix, written in the computational basis, has only (potentially) nonzero entries along the main- and anti-main-diagonals, and zeros elsewhere, hence displaying a matrix form that resembles the alphabet letter `X'.

For an explicit construction in the case of $2\times 3$ systems, let $\{\ket{0},\ket{1}\}\coloneq \{(1\;0)^{\sf T},(0\;1)^{\sf T}\}$ and $\{\ket{0},\ket{1},\ket{2}\}\coloneq \{(1\;0\;0)^{\sf T},(0\;1\;0)^{\sf T},(0\;0\;1)^{\sf T}\}$ be the ordered computational basis sets and vector representations for the qubit and qutrit Hilbert spaces, respectively, where the adopted ket labels were arbitrarily chosen (not to be confused with Fock states). Then the computational basis for a $2\times 3$ system is given by $\{\ket{00}, \ket{01}, \ket{02}, \ket{10}, \ket{11}, \ket{12}\}\coloneq\{\ket{1},\ldots,\ket{6}\}$, where $\ket{ab}\coloneq\ket{a}\otimes\ket{b}$. Thus, in the matrix representation where $\rho_{j,k}\coloneq\langle j|\rho\ket{k}$, an arbitrary $2\times 3$ X-state density matrix can be parametrized as (using dots to represent zeroes)
\begin{equation}\label{eq:qubitqutritXstate}
\arraycolsep=1.38pt
\bm{\rho}_{\rm X}\!=\!\left[\!\begin{array}{cccccc}
a_1 & \cdot & \cdot & \cdot & \cdot & r_1 e^{-i \phi_1}\\
\cdot & a_2 & \cdot & \cdot & r_2 e^{-i \phi_2} & \cdot\\
\cdot & \cdot & a_3 & r_3 e^{-i \phi_3} & \cdot & \cdot\\
\cdot & \cdot & r_3 e^{i \phi_3}  & b_3 & \cdot & \cdot\\
\cdot & r_2 e^{i \phi_2} & \cdot  & \cdot &  b_2 & \cdot\\
r_1 e^{i \phi_1} & \cdot & \cdot  & \cdot &  \cdot & b_1\\
\end{array}\!\right]\!,
\end{equation}
where, for $k=1,2,3$, the parameters $a_k$, $b_k$ and $r_k$ are nonnegative real numbers, and the normalization and positive semidefiniteness of $\bm{\rho}_{\rm X}$ require that
\begin{equation}\label{eq:qualifiers}
\sum_{k=1}^3a_k+b_k=1\quad\mbox{and}\quad r_k\leq\sqrt{a_k b_k}.
\end{equation}

For what follows, it will be useful to characterize the spectrum and the negativity of an arbitrary $2\times 3$ X state. Firstly, a straightforward computation of the eigenvalues of $\bm{\rho}_{\rm X}$ [cf. Eq.~(\ref{eq:qubitqutritXstate})] gives
\begin{equation}\label{eq:lambdas}
\lambda_k^\pm=\frac{a_k+b_k}{2}\pm\sqrt{r_k^2+d_k^2}\quad \mbox{for}\quad k=1,2,3,
\end{equation}
where we have defined 
\begin{equation}
d_k\coloneq \frac{b_k-a_k}{2}.
\end{equation} 

We now seek a general formula for the negativity $\mathcal{N}(\bm{\rho}_{\rm X})$ of $2\times 3$ X states, amenable to the optimization problems to be considered in the forthcoming sections. To start, we can define the negativity of $\bm{\rho}_{\rm X}$ as\footnote{We follow Ref.~\cite{13Hedemann} and define negativity as twice its standard definition, so that the resulting measure of entanglement spans the interval $[0,1]$ in $2\times 3$ systems.}
\begin{equation}
\mathcal{N}(\bm{\rho}_{\rm X})\coloneq\|\bm{\rho}_{\rm X}^{\Gamma}\|_{\text{tr}}-1,
\end{equation}
where $\|\mathbf{A}\|_{\text{tr}}\coloneq \tr[\sqrt{\mathbf{A}^{\dag}\mathbf{A}}]$ is the trace norm of $\mathbf{A}$ and $\bm{\rho}_{\rm X}^{\Gamma}$ is the partial transpose of $\bm{\rho}_{\rm X}$ with respect to either subsystem. Since $\bm{\rho}_{\rm X}$ is Hermitian, we can get $\bm{\rho}_{\rm X}^{\Gamma}$ from~Eq.~(\ref{eq:qubitqutritXstate}) by performing $i\leftrightarrow -i$, $r_1 \leftrightarrow r_3$, and $\phi_1 \leftrightarrow \phi_3$, so its negativity is
\begin{equation}\label{eq:negrhox}
\mathcal{N}(\bm{\rho}_{\rm X})= \left(\,\sum_{k=1}^3|\lambda_k^{\prime +}| + |\lambda_k^{\prime -}|\right) - 1\,,
\end{equation}
where $\{\lambda_k^{\prime\pm}\}_{k=1}^3$ is the set of eigenvalues of $\bm{\rho}_{\rm X}^{\Gamma}$, where
\begin{equation}\label{eq:lambdaprime}
\lambda_k^{\prime\pm}=\frac{a_k+b_k}{2}\pm\sqrt{r_{4-k}^2+d_k^2}\,.
\end{equation}
Now, a few observations regarding the eigenvalues $\{\lambda_k^\pm\}_{k=1}^3$ and $\{\lambda_k^{\prime \pm}\}_{k=1}^3$ are useful:
\begin{itemize}[leftmargin=*]
\item The eigenvalues $\{\lambda_k^{\prime +}\}_{k=1}^3$ are nonnegative [cf. Eq.~(\ref{eq:lambdaprime})], so their absolute values appearing in Eq.~(\ref{eq:negrhox}) can be dropped to give
\begin{equation}\label{eq:negrhox2}
\mathcal{N}(\bm{\rho}_{\rm X})= \left(\,\sum_{k=1}^3\lambda_k^{\prime +} + |\lambda_k^{\prime -}|\right) - 1\,.
\end{equation}
\item The eigenvalues $\lambda_2^{\prime \pm}$ are identical to $\lambda_2^{\pm}$ [cf. Eqs.~(\ref{eq:lambdas}) and (\ref{eq:lambdaprime})] and are thus nonnegative too, so their primes and absolute values appearing in Eq.~(\ref{eq:negrhox2}) can be dropped to give 
\begin{equation}\label{eq:negrhox3}
\mathcal{N}(\bm{\rho}_{\rm X})= \lambda_1^{\prime +} + |\lambda_1^{\prime -}| + \lambda_3^{\prime +} + |\lambda_3^{\prime -}| + \lambda_2^{+} + \lambda_2^{-} - 1\,.
\end{equation}

\item Owing to the normalization condition of $\bm{\rho}_{\rm X}$, the term $\lambda_2^{+} + \lambda_2^{-} - 1$ appearing in Eq.~(\ref{eq:negrhox3}) can be rewritten as the negative of the sum of the remaining eigenvalues of $\bm{\rho}_{\rm X}$, which yields

\begin{equation}\label{eq:negrhox4}
\mathcal{N}(\bm{\rho}_{\rm X}){\kern -0.5pt}={\kern -0.5pt} \lambda_1^{\prime +} {\kern -0.5pt}+{\kern -0.5pt} |\lambda_1^{\prime -}| {\kern -0.5pt}+{\kern -0.5pt} \lambda_3^{\prime +} {\kern -0.5pt}+{\kern -0.5pt} |\lambda_3^{\prime -}| {\kern -0.5pt}-{\kern -0.5pt}\lambda_1^{+} {\kern -0.5pt}-{\kern -0.5pt} \lambda_1^{-} {\kern -0.5pt}-{\kern -0.5pt} \lambda_3^{+} {\kern -0.5pt}-{\kern -0.5pt} \lambda_3^{-} .
\end{equation}

\item As demonstrated in App.~\ref{app:onenegeigenvalue}, $2\times 3$ X states have at most one negative eigenvalue in their partial transpose; either $\lambda_1^{\prime -}$ or $\lambda_3^{\prime -}$. Without loss of generality, henceforth we assume $\bm{\rho}_{\rm X}$ to be an entangled X state with $\lambda_1^{\prime -}<0$ and $\lambda_3^{\prime -}\geq 0$ and, accordingly, we rewrite Eq.~(\ref{eq:negrhox4}) as
\begin{equation}\label{eq:negrhox5}
\mathcal{N}(\bm{\rho}_{\rm X})= \lambda_1^{\prime +} - \lambda_1^{\prime -} + \lambda_3^{\prime +} + \lambda_3^{\prime -} -\lambda_1^{+} - \lambda_1^{-} - \lambda_3^{+} - \lambda_3^{-} .
\end{equation}
\item From Eqs.~(\ref{eq:lambdas}) and (\ref{eq:lambdaprime}), notice that $\lambda_3^{\prime +} + \lambda_3^{\prime -} = \lambda_3^{+} + \lambda_3^{-}$, so Eq.~(\ref{eq:negrhox5}) simplifies to
\begin{equation}\label{eq:negrhox6}
\mathcal{N}(\bm{\rho}_{\rm X})= \lambda_1^{\prime +} - \lambda_1^{\prime -}  -\lambda_1^{+} - \lambda_1^{-} .
\end{equation}
\item Equation~(\ref{eq:lambdaprime}) implies that
\begin{equation}\label{eq:lambda1primediff}
\lambda_1^{\prime +}-\lambda_1^{\prime -} = 2\sqrt{r_3^2+d_1^2}\,,
\end{equation}
and, according to Eq.~(\ref{eq:lambdas}), the quantities $r_3^2$ and $d_1^2$ admit the expressions
\begin{align}
r_3^2&=\left(\frac{\lambda_3^+-\lambda_3^-}{2}\right)^2-d_3^2\,,\\
 d_1^2&=\left(\frac{\lambda_1^+-\lambda_1^-}{2}\right)^2-r_1^2\,,
\end{align}
which, when plugged back into Eq.~(\ref{eq:lambda1primediff}), lead to the final form for $\mathcal{N}(\bm{\rho}_{\rm X})$ as
\begin{widetext}
\begin{equation}\label{eq:Nrhoxfinal}
\mathcal{N}(\bm{\rho}_{\rm X})=-\lambda_1^+-\lambda_1^-+\sqrt{(\lambda_1^+-\lambda_1^-)^2+(\lambda_3^+-\lambda_3^-)^2-4d_3^2-4r_1^2}.
\end{equation}
\end{widetext}
\end{itemize}
Equation~(\ref{eq:Nrhoxfinal}) will be the starting point for the maximization of the negativity of entanglement of a $2\times 3$ X state of fixed spectrum, to be carried out in Sec.~\ref{sec:XMEMSspectrum}.

\subsection{Semidefinite Programming}\label{sec:SDP}

A semidefinite program (SDP) is an optimization problem that admits the following \emph{inequality form}:
\begin{equation}\label{eq:SDPineqform}
\mbox{minimize }\bm{c}^{\sf T}\bm{x}\quad\mbox{such that } \bm{F}(\bm{x})\coloneq \bm{F}_0+\sum_{i=1}^n \bm{F}_i x_i \geq 0\,,
\end{equation}
where $\bm{c}\in\mathbb{R}^n$ is a given vector and $\{\bm{F}_i\}_{i=0}^n$ are given $d$-dimensional Hermitian matrices. The problem variables are $\{x_i\}_{i=1}^n$, which form the variable vector $\bm{x}$. If $\{\bm{F}_i\}_{i=0}^n$ are diagonal matrices, then the requirement $\bm{F}(\bm{x})\geq 0$ is equivalent to the set of linear inequalities $[\bm{F}(\bm{x})]_{jj}\geq 0$, for $j=1,\ldots,d$. In this particular case, SDPs reduce to simple linear programs. If, on the other hand, not all $\{\bm{F}_i\}_{i=0}^n$ are diagonal, then the \emph{linear matrix inequality} (LMI) $\bm{F}(\bm{x})\geq 0$ expresses the requirement that $\bm{F}(\bm{x})$ be positive semidefinite (hence, the term SDP), being thus equivalent to a set of nonlinear inequalities that guarantee the nonnegativity of each eigenvalue of $\bm{F}(\bm{x})$ [e.g., those enforcing all principal minors of $\bm{F}(\bm{x})$ to be nonnegative]. An SDP is, in this sense, a (nonlinear) generalization of a linear program.

Despite the implicit nonlinearity in the LMI, SDPs are still convex optimization problems because the solution set of $\bm{F}(\bm{x})\geq 0$ is convex~\cite{04Boyd} (and the objective function is linear). As a result, any local optimum of an SDP is automatically a global optimum. SDPs thus constitute a special family of nonlinear optimization problems that can be efficiently solved.

An important counterpart of problem (\ref{eq:SDPineqform}) is the optimization problem,
\begin{equation}\label{eq:SDPdual}
\mbox{maximize }-\tr[\bm{F}_0 \bm{Z}]\quad\mbox{such that }
\left\{\begin{array}{rcl}
 \bm{Z}&\geq& 0\,,\\ 
 \tr[\bm{F}_i\bm{Z}]&=&c_i\quad \forall i\,,
\end{array}\right.
\end{equation}
to be solved for the matrix variable $\bm{Z}$. Problem~(\ref{eq:SDPdual}) is the \emph{dual problem} of problem (\ref{eq:SDPineqform}) which, in turn, is referred to as the \emph{primal problem}. Remarkably, by expanding $\bm{Z}$ in a Hermitian basis and solving the linear system  $\{\tr[\bm{F}_i\bm{Z}]=c_i\}_{i=1}^n$ for the relevant expansion coefficients, the determination of the remaining coefficients can be cast as an optimization problem in the inequality form, meaning that the dual of an SDP is another SDP.

An important property to be explored throughout is that the solution of the dual problem sets a lower bound for the solution of the primal problem. In order to see that, let $\mathfrak{p}^\ast$ and $\mathfrak{d}^\ast$ denote the \emph{optimal values} of the primal and dual objective functions, respectively. Then
\begin{eqnarray}
\mathfrak{d}^\ast-\mathfrak{p}^\ast &=& -\tr[\bm{F}_0 \bm{Z}_{\text{opt}}]-\bm{c}^{\sf{T}}\bm{x}_{\text{opt}}\nonumber\\
&=&-\tr[\bm{F}_0 \bm{Z}_{\text{opt}}]-\sum_{i=1}^n\tr[\bm{F}_i \bm{Z}_{\text{opt}}]x_{i,\text{opt}}\nonumber\\
&=& -\tr\left[\left(\bm{F}_0 + \sum_{i=1}^n\bm{F}_i x_{i,\text{opt}} \right) \bm{Z}_{\text{opt}}\right]\nonumber\\
&=&-\tr\left[\bm{F}(\bm{x}_{\text{opt}}) \bm{Z}_{\text{opt}}\right],
\end{eqnarray}
where $\bm{x}_{\text{opt}}$ and $\bm{Z}_{\text{opt}}$ denote optimal choices for the variables $\bm{x}$ and $\bm{Z}$, respectively. Since both $\bm{F}(\bm{x}_{\text{opt}})$ and $\bm{Z}_{\text{opt}}$ are required to be nonnegative [cf. the constraints of the primal and dual problems], it follows that $\mathfrak{d}^\ast\leq\mathfrak{p}^\ast$. Moreover, since the primal problem is a minimization and the dual problem is a maximization, we arrive at the \emph{weak duality} property for SDPs,
\begin{equation}\label{eq:weakduality}
\mathfrak{d}\leq\mathfrak{d}^\ast\leq\mathfrak{p}^\ast\leq\mathfrak{p}\,,
\end{equation}
where $\mathfrak{d}$ and $\mathfrak{p}$ denote values for the dual and primal objective functions computed from choices of $\bm{Z}$ and $\bm{x}$ that fulfill the problem's constraints, but are not required to be optimal ones. In contrast with the optimal values $\mathfrak{d}^\ast$ and $\mathfrak{p}^\ast$, the numbers $\mathfrak{d}$ and $\mathfrak{p}$ are commonly referred to as \emph{feasible values}.

Inequality~(\ref{eq:weakduality}) enables a handy strategy to confirm the optimality of a candidate solution of an SDP: Suppose we know \emph{feasible points} $\widetilde{\bm{Z}}$ and $\widetilde{\bm{x}}$ for which the corresponding feasible values satisfy $\mathfrak{d}=\mathfrak{p}$. Then, according to (\ref{eq:weakduality}), $\mathfrak{p}^\ast=\mathfrak{p}$ also, which implies that $\widetilde{\bm{x}}$ is actually an optimal point of the primal problem. Of course, it may be difficult to find dual and primal feasible points such that $\mathfrak{d}=\mathfrak{p}$, but to that purpose one can rely upon all sorts of unorthodox methods to approach the dual and primal optimization problems (e.g., input from numerical experiments and educated guesses). If, at the end, the identity $\mathfrak{d}=\mathfrak{p}$ holds, then all the dubious steps can be forgotten and the optimality can be rigorously claimed. As we shall see, three important theorems stated in this paper have been proved in this way.

\section{$\mathbf{2\times 3}$ X-MEMS wrt spectrum}\label{sec:XMEMSspectrum}

What is the $2\times 3$ X state of a given (but arbitrary) spectrum $\Lambda$ that reaches maximal negativity of entanglement? This section is devoted to answer this question and, in so doing, to establish a family of $2\times 3$ X-MEMS wrt spectrum.

Formally, the aforementioned question can be cast as the optimization problem

\begin{equation}\label{eq:optspectrum1}
\mbox{maximize }\mathcal{N}(\bm{\rho}_{\rm X})\quad\mbox{such that } \{\lambda_k^\pm\}_{k=1}^3=\Lambda,
\end{equation}
where $\Lambda\coloneq\{\lambda_j\}_{j=1}^6$ denotes any chosen spectrum with elements $\lambda_1\geq\lambda_2\geq\lambda_3\geq\lambda_4\geq\lambda_5\geq\lambda_6$. Moreover, the figure-of-merit $\mathcal{N}(\bm{\rho}_{\rm X})$ is given by Eq.~(\ref{eq:Nrhoxfinal}) and the maximization runs over the set of variables $\{\lambda_k^\pm\}_{k=1}^3\cup\{d_3,r_1\}$. Admittedly, a few constraints have been omitted from the optimization problem~(\ref{eq:optspectrum1}), but as we now argue, they do not influence the solution of the problem. 

First, we should have imposed that $\lambda_k^+\geq\lambda_k^-$ for every $k\in\{1,2,3\}$, as implied by Eq.~(\ref{eq:lambdas}). However, $\mathcal{N}(\bm{\rho}_{\rm X})$ is invariant under exchanges $\lambda_k^+\leftrightarrow\lambda_k^-$ [cf. Eq.~(\ref{eq:Nrhoxfinal})], so the inequality constraints $\lambda_k^+\geq\lambda_k^-$ can be ignored \emph{a priori} and accommodated \emph{a posteriori}. Secondly, the variable subset $\{d_3,r_1\}$ is not independent from the variable subset $\{\lambda_k^\pm\}_{k=1}^3$ [cf. Eq.~(\ref{eq:lambdas})], thus their exact dependencies should have appeared as additional constraints in the optimization problem~(\ref{eq:optspectrum1}). By relaxing such constraints (i.e., regarding $d_3$ and $r_1$ as independent variables), the optimal value arising from problem~(\ref{eq:optspectrum1}) will be, in principle, an upper bound on the maximal negativity associated with the spectrum $\Lambda$, which we call $N_{{\rm X},\Lambda}$. We defer until the end of this section an explicit verification that $N_{{\rm X},\Lambda}$ is, however, a \emph{tight} upper bound.

Regarding $d_3$ and $r_1$ as independent variables, from Eq.~(\ref{eq:Nrhoxfinal}) we can see that $\mathcal{N}(\bm{\rho}_{\rm X})$ is maximized if $d_3=r_1=0$, in which case problem~(\ref{eq:optspectrum1}) reduces to

\begin{align}
\mbox{maximize}&\quad -\lambda_1^+-\lambda_1^-+\sqrt{\left(\lambda_1^+-\lambda_1^-\right)^2+\left(\lambda_3^+-\lambda_3^-\right)^2}\nonumber\\
\mbox{such that}&\quad \{\lambda_k^\pm\}_{k=1}^3=\{\lambda_j\}_{j=1}^6\,. \label{eq:proboptdist}
\end{align}

Applying Proposition~\ref{proposition1} (stated and proved in App.~\ref{app:discopt}), the optimal solution for this problem (complying with $\lambda_k^+\geq \lambda_k^-$) can be promptly found to be
\begin{align}\label{eq:onlyoptsol}
\lambda_1^+=\lambda_4\,,&\quad \lambda_1^-=\lambda_6\,,\nonumber\\
\lambda_2^+=\lambda_2\,,&\quad \lambda_2^-=\lambda_3\,,\\
\lambda_3^+=\lambda_1\,,&\quad \lambda_3^-=\lambda_5\,,\nonumber
\end{align}
thus establishing 
\begin{equation}\label{eq:defNLambdaub}
\mathcal{N}(\bm{\rho}_{\rm X})\leq N_{{\rm X},\Lambda}=-\lambda_4-\lambda_6+\!\sqrt{(\lambda_4-\lambda_6)^2+(\lambda_1-\lambda_5)^2}
\end{equation}
for an arbitrary $2\times 3$ X state $\bm{\rho}_{\rm X}$ with spectrum $\Lambda$.

In order to see that $N_{{\rm X},\Lambda}$ is a tight upper bound,  consider the following X state of spectrum $\Lambda$,

\begin{equation}\label{eq:egXMEMSwrtspec}
\bm{\varrho}_{\rm X}=\frac{1}{2}\left[
\begin{array}{cccccc}
2\lambda_4 & \cdot & \cdot & \cdot & \cdot & \cdot \\
\cdot & 2\lambda_2 & \cdot & \cdot & \cdot & \cdot \\
\cdot & \cdot & \lambda_1+\lambda_5 & \lambda_1-\lambda_5 & \cdot & \cdot\\
\cdot & \cdot & \lambda_1-\lambda_5 & \lambda_1+\lambda_5 & \cdot & \cdot\\
\cdot & \cdot & \cdot & \cdot & 2\lambda_3 & \cdot\\
\cdot & \cdot & \cdot & \cdot & \cdot & 2\lambda_6\\
\end{array}
\right]\!,
\end{equation}
whose negativity can be written as
\begin{equation}
\mathcal{N}(\bm{\varrho}_{\rm X})=\max\left[0,N_{{\rm X},\Lambda}\right]\,.
\end{equation}
Unless $N_{{\rm X},\Lambda}\leq 0$ (which matches the PPT condition~\cite{96Peres1413,96Horodecki1} and the condition for ``separability from spectrum'' for $2\times 3$ states~\cite{07Hildebrand052325}), the state $\bm{\varrho}_{\rm X}$ is entangled and has negativity $N_{{\rm X},\Lambda}$.

To sum up, we have established $N_{{\rm X},\Lambda}$ [cf. Eq.~(\ref{eq:defNLambdaub})] as the maximal negativity attainable by any $2\times 3$ X state of spectrum $\Lambda$, and  $\bm{\varrho}_{\rm X}$ [cf. Eq.~(\ref{eq:egXMEMSwrtspec})] as a family of $2\times 3$ X states --- parametrized by their spectrum $\Lambda$ --- that reach the maximal negativity $N_{{\rm X},\Lambda}$. Therefore, $\bm{\varrho}_{\rm X}$  represents a realization of the so-called $2\times 3$ X-MEMS wrt spectrum.

\section{Spectrum-Constrained $\mathbf{2\times 3}$ X-MEMS wrt purity}\label{sec:spconstrainedXMEMS}

Since a spectrum determines a value of purity $P$, every X-MEMS wrt $P$ must be an X-MEMS wrt a spectrum that realizes that $P$. Therefore, X-MEMS wrt to $P$ can be found by searching strictly on the set of X-MEMS wrt spectra that realize $P$, by solving the optimization problem 
\begin{align}
\mbox{maximize}&\quad -\lambda_4-\lambda_6+\sqrt{\left(\lambda_4-\lambda_6\right)^2+\left(\lambda_1-\lambda_5\right)^2}\nonumber\\
\mbox{such that}&\quad \lambda_{6}\geq 0\,,\quad\sum_{i=1}^{6}\lambda_i=1\,,\quad \sum_{i=1}^{6}\lambda_i^2\leq P\,,\label{eq:optgeneral}
\end{align}
on the \emph{variables} $\lambda_1\geq\lambda_2\geq\lambda_3\geq\lambda_4\geq\lambda_5\geq\lambda_6\geq 0$ and for a fixed $P\in\,]\frac{1}{5},1[$. 

A few observations concerning the optimization problem~(\ref{eq:optgeneral}) are relevant. First, although $2\times 3$ mixed states span the purity interval $[1/6,1[$, we disregard the range $[1/6,1/5]$ because these highly mixed states are known to be separable~\cite{98Zyczkowski883}. Second, there is no need to further constrain problem~(\ref{eq:optgeneral}) with $\lambda_1\geq\lambda_2\geq\lambda_3\geq\lambda_4\geq\lambda_5\geq\lambda_6$ because, owing to the nature of the figure-of-merit (and to the invariance of the last two constraints under $\lambda_i\leftrightarrow\lambda_j$), this ordering will be fulfilled even unimposed (cf. Proposition~\ref{proposition1}). Third, rather than the physically motivated purity constraint $\sum_{i=1}^{6}\lambda_i^2= P$, we have required $\sum_{i=1}^{6}\lambda_i^2\leq P$. Since the maximum of a convex function relative to a convex set generally occurs at some extreme point~\cite[Chapter 32]{70TyrrelRockafellar}, enlarging the feasible region to its convex hull does not affect the solution of the problem and brings technical advantages towards finding its optimal solution.

Despite the convexity of the feasible set, problem (\ref{eq:optgeneral}) is not a convex optimization problem because it aims to \emph{maximize} (as opposed to minimize) a convex objective function. In this section, we avoid the difficulties of looking for global optima in nonconvex problems and, instead, introduce three extra sets of spectral constraints that render the objective function linear and the resulting optimization problems convex\footnote{The maximization of a concave objective function over a convex set is a convex optimization problem. By making our objective function linear, it becomes simultaneously convex and concave, so that the resulting problem is a special type of convex optimization problem.}; namely (i) $\lambda_3=\lambda_4=\lambda_5=\lambda_6=0$, (ii) $\lambda_4=\lambda_5=\lambda_6=0$ and (iii) $\lambda_4=\lambda_5=\lambda_6$. Due to these added constraints, the solutions of the resulting optimization problems do not necessarily produce maximally entangled mixed states of purity $P$ amongst all $2\times 3$ X states. Instead, they yield maximally entangled mixed X states of purity $P$ \emph{and} (i) rank 2, (ii) rank 3, and (iii) three-fold degeneracy of the smallest eigenvalue, respectively. For now, we shall content ourselves with this limited scenario inasmuch as the added constraints are physically meaningful and the resulting states are of practical and theoretical interest. 

In what follows, we state three theorems that analytically characterize the three aforementioned families of maximal X states. The proofs are in App.~\ref{app:XMEMSwrtP}.

\begin{theorem}\label{thm:rank2} The maximal negativity achievable by a $2\times 3$ rank-2 X state of purity $P~\in~[\frac{1}{2},1[$ is 
\begin{equation}\label{eq:NXP2}
N_{{\rm X},P}^{(2)}=\frac{1}{2}(1+f_P),
\end{equation}
where
\begin{equation}\label{eq:fP}
f_P\coloneq\sqrt{2P-1}\,.
\end{equation} 
The density matrix
\begin{equation}\label{eq:rhoXP2}
\bm{\varrho}_{{\rm X},P}^{(2)}=
\frac{1}{2}\left[
\begin{array}{cccccc}
\cdot & \cdot & \cdot & \cdot & \cdot & \cdot \\
\cdot & 2\widetilde{\lambda}_2 & \cdot & \cdot & \cdot & \cdot \\
\cdot & \cdot & \widetilde{\lambda}_1 & \widetilde{\lambda}_1 & \cdot & \cdot\\
\cdot & \cdot & \widetilde{\lambda}_1 & \widetilde{\lambda}_1 & \cdot & \cdot\\
\cdot & \cdot & \cdot & \cdot & \cdot & \cdot\\
\cdot & \cdot & \cdot & \cdot & \cdot & \cdot\\
\end{array}
\right]\!,
\end{equation}
with 
\begin{equation}\label{eq:lambda1lambda2rank2}
\widetilde{\lambda}_1=\frac{1}{2}(1+f_P)\quad\mbox{and}\quad\widetilde{\lambda}_2=\frac{1}{2}(1-f_P),
\end{equation}
gives a construction of such a maximal state; that is $\mathcal{N}[\bm{\varrho}_{{\rm X},P}^{(2)}]=N_{{\rm X},P}^{(2)}$.
\end{theorem}

\begin{theorem}\label{thm:rank3} The maximal negativity achievable by a $2\times 3$ rank-3 X state of purity $P~\in~[\frac{1}{3},1[$ is 
\begin{equation}\label{eq:NXP3}
N_{{\rm X},P}^{(3)}=\frac{1}{3}(1+g_P),
\end{equation}
where
\begin{equation}
g_P\coloneq\sqrt{6P-2}\,.
\end{equation} 
The density matrix
\begin{equation}\label{eq:rhoXP3}
\bm{\varrho}_{{\rm X},P}^{(3)}=
\frac{1}{2}\left[
\begin{array}{cccccc}
\cdot & \cdot & \cdot & \cdot & \cdot & \cdot \\
\cdot & 2\widetilde{\lambda}_2 & \cdot & \cdot & \cdot & \cdot \\
\cdot & \cdot & \widetilde{\lambda}_1 & \widetilde{\lambda}_1 & \cdot & \cdot\\
\cdot & \cdot & \widetilde{\lambda}_1 & \widetilde{\lambda}_1 & \cdot & \cdot\\
\cdot & \cdot & \cdot & \cdot & 2\widetilde{\lambda}_3 & \cdot\\
\cdot & \cdot & \cdot & \cdot & \cdot & \cdot\\
\end{array}
\right]\!,
\end{equation}
with 
\begin{equation}\label{eq:lambda1lambda2rank3}
\widetilde{\lambda}_1=\frac{1}{3}(1+g_P)\quad\mbox{and}\quad\widetilde{\lambda}_2=\widetilde{\lambda}_3=\frac{1}{6}(2-g_P),
\end{equation}
gives a construction of such a maximal state; that is $\mathcal{N}[\bm{\varrho}_{{\rm X},P}^{(3)}]=N_{{\rm X},P}^{(3)}$.
\end{theorem}

\begin{theorem}\label{thm:deg} The maximal negativity achievable by a $2\times 3$ X state of purity $P~\in~]\frac{1}{5},1[$ with triple degeneracy of the smallest eigenvalue is 
\begin{equation}\label{eq:NXPdeg}
N_{{\rm X},P}^{(\mathrm{deg})}=\left\{\begin{array}{rr}
\frac{1}{3}(-1+5h_P)&\mbox{for } P \in ]\frac{1}{5},\frac{3}{8}[\\
\frac{1}{3}(1+g_P)&\mbox{for } P \in [\frac{3}{8},1[
\end{array}\right.,
\end{equation}
where
\begin{equation}\label{eq:defgh}
g_P\coloneq\sqrt{6P-2}\quad\mbox{and}\quad h_P\coloneq\sqrt{\frac{6P}{5}-\frac{1}{5}}\,.
\end{equation} 
The density matrix
\begin{equation}\label{eq:rhoXPdeg}
\bm{\varrho}_{{\rm X},P}^{(\mathrm{deg})}=
\frac{1}{2}\left[
\begin{array}{cccccc}
2\widetilde{\lambda} & \cdot & \cdot & \cdot & \cdot & \cdot \\
\cdot & 2\widetilde{\lambda}_2 & \cdot & \cdot & \cdot & \cdot \\
\cdot & \cdot & \widetilde{\lambda}_1+\widetilde{\lambda} & \widetilde{\lambda}_1-\widetilde{\lambda} & \cdot & \cdot\\
\cdot & \cdot & \widetilde{\lambda}_1-\widetilde{\lambda} & \widetilde{\lambda}_1+\widetilde{\lambda} & \cdot & \cdot\\
\cdot & \cdot & \cdot & \cdot & 2\widetilde{\lambda}_3 & \cdot\\
\cdot & \cdot & \cdot & \cdot & \cdot & 2\widetilde{\lambda}\\
\end{array}
\right]\!,
\end{equation}
with 
\begin{align}\label{eq:lambdasdeg}
\widetilde{\lambda}_1&=\left\{\begin{array}{rr}
\frac{1}{6}(1+4h_P)&\mbox{for } P \in ]\frac{1}{5},\frac{3}{8}[\\
\frac{1}{3}(1+g_P)&\mbox{for } P \in [\frac{3}{8},1[
\end{array}\right.\nonumber\\
\widetilde{\lambda}_2=\widetilde{\lambda}_3&=\left\{\begin{array}{rr}
\frac{1}{6}(1+h_P)&\mbox{for } P \in ]\frac{1}{5},\frac{3}{8}[\\
\frac{1}{6}(2-g_P)&\mbox{for } P \in [\frac{3}{8},1[
\end{array}\right.\\
\widetilde{\lambda}&=\left\{\begin{array}{rr}
\frac{1}{6}(1-2h_P)&\mbox{for } P \in ]\frac{1}{5},\frac{3}{8}[\\
0&\mbox{for } P \in [\frac{3}{8},1[
\end{array}\right. ,\nonumber
\end{align}
%
%
gives a construction of such a maximal state; that is $\mathcal{N}[\bm{\varrho}_{{\rm X},P}^{(\mathrm{deg})}]=N_{{\rm X},P}^{(\mathrm{deg})}$.
\end{theorem}

As it should be expected, the more stringent the added constraints are, the lower the maximal attainable negativity will be. In fact, for the applicable values of $P$, we generally have
\begin{equation}\label{eq:hierarchy}
N_{{\rm X},P}^{(\text{deg})}\geq N_{{\rm X},P}^{(3)} \geq N_{{\rm X},P}^{(2)}\,,
\end{equation}
as depicted in Fig.~\ref{fig:NXP}.  Remarkably, $N_{{\rm X},P}^{(\text{deg})}$ equals $N_{{\rm X},P}^{(3)}$ for $P\in[3/8,1[$, as $\bm{\varrho}_{{\rm X},P}^{(\text{deg})}$ turns out to be rank-3 in this purity domain [cf. Eq.~(\ref{eq:lambdasdeg})]. Nevertheless, the possibility of making $\bm{\varrho}_{{\rm X},P}^{(\text{deg})}$ full-rank (yet three-fold degenerate; $\lambda_4=\lambda_5=\lambda_6$)  is generally benign; $N_{{\rm X},P}^{(\text{deg})}$ is slightly greater than $N_{{\rm X},P}^{(3)}$ for $P\in\,[1/3,3/8[$, as noticeable from the magnified inset in Fig.~\ref{fig:NXP}.

\begin{figure}[h]
\centering
\includegraphics[width=\columnwidth]{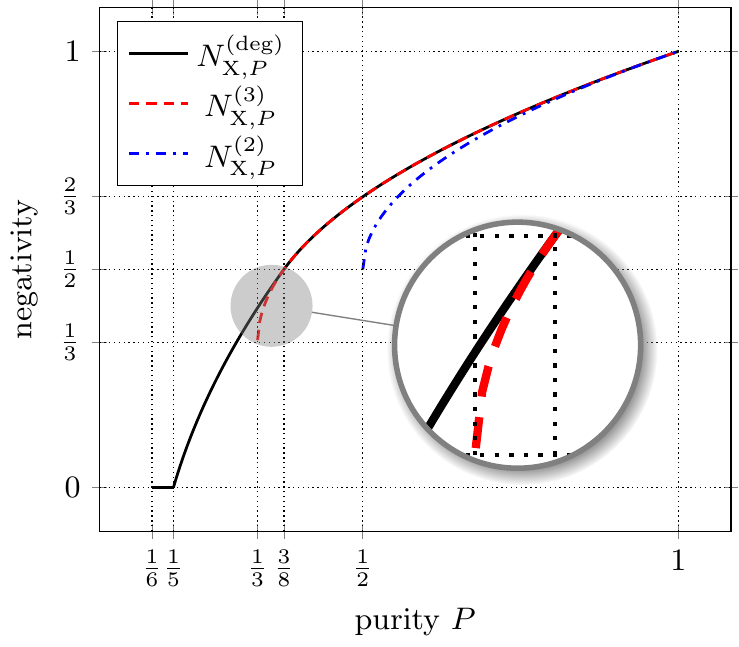}
\caption{(Color online) Maximal negativity attainable by X states of purity $P$ and (i) rank-2 [$N_{{\rm X},P}^{(2)}$], (ii) rank-3 [$N_{{\rm X},P}^{(3)}$] and (iii) triple degeneracy of the smallest eigenvalue [$N_{{\rm X},P}^{(\text{deg})}$]. In the purity range $P\in[3/8,1[$, the identity between $N_{{\rm X},P}^{(\text{deg})}$ and $N_{{\rm X},P}^{(3)}$ holds, meaning that it is not possible to increase the negativity by raising the rank from $3$ while keeping the smallest eigenvalue three-fold degenerate. Nevertheless, this proves to be a fruitful strategy in the purity range $P\in]1/5,3/8[$, where $N_{{\rm X},P}^{(\text{deg})}$ turns out to be slightly greater than $N_{{\rm X},P}^{(3)}$ (cf. magnified inset).}\label{fig:NXP}
\end{figure}

It is instructive to compare the results of this section with the candidate $2\times 3$ MEMS proposed in Ref.~\cite{13Hedemann}, whose negativity in terms of $P$ can be shown to be
\begin{equation}\label{eq:nfunp}
N_{P}^{(\text{Hed})}\coloneq\left\{\begin{array}{l}
\frac{1}{5}\left[-1+e_P+\sqrt{(-1+e_P)^2-\frac{25}{4} e_P^2}\right]\\
\hspace{4.3cm}\mbox{for } P \in ]\frac{1}{5},\frac{3}{8}[ \,,\\[2mm]
\frac{1}{3}(1+g_P)\quad \mbox{for } P \in [\frac{3}{8},1[\,,
\end{array}
\right.
\end{equation}
where $g_P$ has already been defined in Eq.~(\ref{eq:defgh}) and
\begin{equation}
e_P\coloneq \sqrt{\frac{40P}{7}-\frac{8}{7}}.
\end{equation}
Remarkably, for $P\in[\frac{3}{8},1[$, $N_P^{(\text{Hed})}$ matches $N_{{\rm X},P}^{(\text{deg})}$ [cf. Eq.~(\ref{eq:NXPdeg})]. However, for $P\in]\frac{1}{5},\frac{3}{8}[$, $N_{{\rm X},P}^{(\text{deg})}$ is slightly greater than $N_P^{(\text{Hed})}$. Although the difference $N_{{\rm X},P}^{(\text{deg})}-N_{P}^{(\text{Hed})}$ is very small, as shown in Fig.~\ref{fig:disqualify}, it suffices to disqualify the prototype state of Ref.~\cite[Eq.~(64)]{13Hedemann} as an actual representation of a $2\times 3$ MEMS.

\begin{figure}[h]
\centering
\includegraphics[width=\columnwidth]{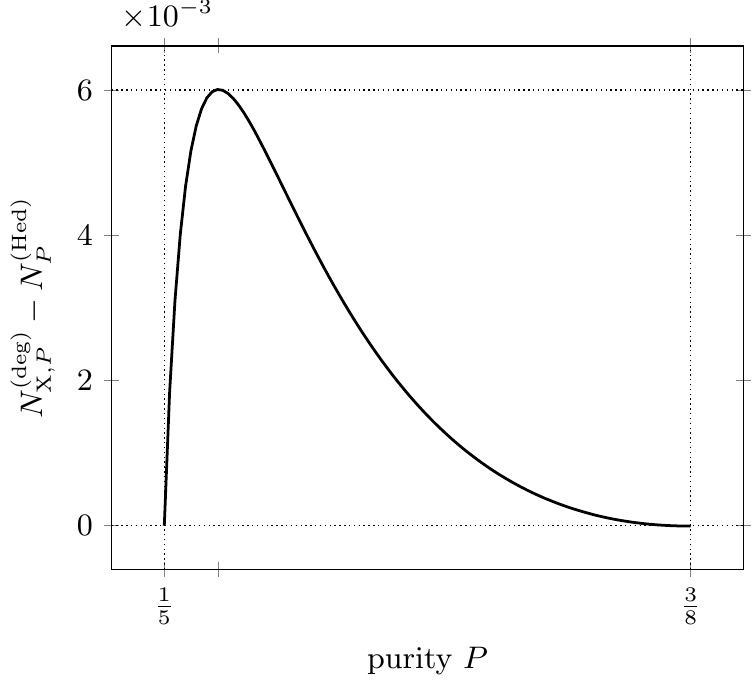}
\caption{Difference between $N_{{\rm X},P}^{(\text{deg})}$ and $N_{{\rm X},P}^{(\text{Hed})}$ --- the conjectured maximal negativity in terms of purity for $2\times 3$ states according to Ref.~\cite{13Hedemann}. The fact that $N_{{\rm X},P}^{(\text{deg})}>N_P^{(\text{Hed})}$ for $P\in]\frac{1}{5},\frac{3}{8}[$ refutes the conjecture that the prototype states of Ref.~\cite[Eq.~(64)]{13Hedemann} are actual $2\times 3$ MEMS wrt purity.}\label{fig:disqualify}
\end{figure}


         

In the next section, we explore the results of Theorems~\ref{thm:rank2} and~\ref{thm:rank3} to show that, unlike it occurs for two-qubit X states~\cite{13Hedemann,14Mendonca79}, it is not generally possible to map an arbitrary $2\times 3$ state into a $2\times 3$ X state of same negativity and purity by means of a unitary transformation, this time supporting the conjecture in Ref.~\cite{13Hedemann}.

\section{Entanglement Universality of $\mathbf{2\times 3}$ X states}\label{sec:universality}

In Ref.~\cite{13Hedemann}, Hedemann proposed that it may be possible to \emph{unitarily} transform any two-qubit state into an X state of the same entanglement; a conjecture that was later confirmed in Ref.~\cite{14Mendonca79} (with entanglement measured by concurrence, negativity, or relative entropy of entanglement). Furthermore, Ref.~\cite{13Hedemann} also conjectured that the property of entanglement-preserving unitary (EPU) equivalence to X states would not hold for systems larger than two qubits, and proposed the true-generalized X states (TGX states) as a nonX alternative that \textit{does} achieve EPU equivalence. In this section, we confirm the conjecture of Ref.~\cite{13Hedemann} that EPU equivalence of two-qubit X states is not inherited by $2\times 3$ systems, that is, \emph{there are $2\times 3$ states which cannot be unitarily mapped into $2\times 3$ X states of the same entanglement}. 

In order to establish this negative result, note that if EPU equivalence to X states \textit{were} to hold, then for every $2\times 3$ (input) state there would be a $2\times 3$ (output) X state of same rank, purity, and entanglement (since rank and purity are preserved under unitary transformations). So, by contraposition, if we can find even just \textit{one} $2\times 3$ \textit{non}X state for which there is \textit{not} a corresponding $2\times 3$ X state of same rank, purity, and negativity, then EPU equivalence cannot hold for $2\times 3$ X states. 

Next, we demonstrate that there exist rank-2 nonX states with purity $P$ that exceed the negativity threshold $N_{{\rm X},P}^{(2)}$. Since this is the maximal negativity achievable by a $2\times 3$ X state of purity $P$ and rank 2 (cf. Theorem~\ref{thm:rank2}), there are no unitarily relatable X-counterparts of the same negativity for any of the constructed states, and thus the rank-2 X states are not EPU equivalent to general $2\times 3$ states, in agreement with~\cite{13Hedemann}.

Consider the following parametric family of rank-2 \textit{non}X states introduced in~\cite{13Hedemann} as a $2\times 3$ rank-2 candidate for \emph{true-generalized X states}\footnote{The ``true generalization'' here implies a valid extension of the EPU equivalence of two-qubit X states to higher dimensional systems~\cite{13Hedemann}; namely, every state can be unitarily mapped into a corresponding TGX state of the same entanglement. However, whether or not the state of Eq.~(\ref{eq:rhoTGX2}) truly parametrizes a $2\times 3$ rank-2 TGX state remains an open problem. Throughout, we stick to the name nonetheless.} (TGX states),
\begin{equation}\label{eq:rhoTGX2}
\bm{\varrho}^{(2)}_{\text{TGX}} =\frac{1}{2} \left( {\begin{array}{cccccc}
   2 p_1 c_{\theta_1}^2 &  \cdot  &  \cdot  &  \cdot  &  \cdot  & p_1 s_{2\theta _1}\\
    \cdot  & 2 p_2 c_{\theta _2 }^2  & \cdot & p_2 s_{2\theta_2}  &  \cdot  & \cdot   \\
    \cdot  &  \cdot  &  \cdot  &  \cdot  &  \cdot  &  \cdot   \\
    \cdot  & p_2 s_{2\theta_2} &  \cdot  & 2 p_2 s_{\theta _2}^2 &  \cdot  & \cdot \\
    \cdot  &  \cdot  &  \cdot  &  \cdot  &  \cdot  &  \cdot   \\
   p_1 s_{2\theta _1} &  \cdot  &  \cdot  &  \cdot  &  \cdot  & 2 p_1 s_{\theta _1 }^2 \\
\end{array}} \right)\!,
\end{equation}
where $c_x  \coloneq \cos (x)$ and $s_x  \coloneq \sin (x)$. The only nonzero eigenvalues of $\bm{\varrho}^{(2)}_{\text{TGX}}$ are parametrically given by the probabilities $p_1$ and $p_2$ ($p_{1},p_{2}> 0$ and $p_1+p_2=1$), whereas its purity is $P =p_1^2+p_2^2$, and its negativity $\mathcal{N}_{\text{TGX}}^{(2)}=\mathcal{N}_{\text{TGX}}^{(2)}(\theta_1,\theta_2,p_1,p_2)$ is given by
\begin{align}
   \mathcal{N}_{\text{TGX}}^{(2)}=&  - p_1c_{\theta_1 }^2  - p_2 s_{\theta _2 }^2 + \sqrt {p_1^2c_{\theta _1 }^4  + p_2^2s_{2\theta _2}^2 }  \nonumber\\
&+ \sqrt {p_2^2s_{\theta _2}^4  + p_1^2s_{2\theta _1}^2 }\,.
\end{align}
For fixed values of purity $P\in[1/2,1[$, which determine the probabilities $p_1=(1+f_P)/2$ and $p_2=(1-f_P)/2$ where $f_P$ is from Eq.~(\ref{eq:fP}), we ran unconstrained numerical maximizations\footnote{We employ the MATLAB function \texttt{fminunc}, running the active-set algorithm with random initial guesses and termination tolerances for the parameter values (\texttt{TolX}) and objective function value (\texttt{TolFun}) set to $10^{-6}$.} of $\mathcal{N}_{\text{TGX}}^{(2)}$ over $\theta_1$ and $\theta_2$, the results of which are denoted by $N_{\text{TGX},P}^{(2)}$. These are shown in Fig.~\ref{fig:TGXrank2} along with $N_{{\rm X},P}^{(2)}$ and $N_{{\rm X},P}^{(\text{deg})}$ for the sake of comparison.

\begin{figure}[h]
\centering
\includegraphics[width=\columnwidth]{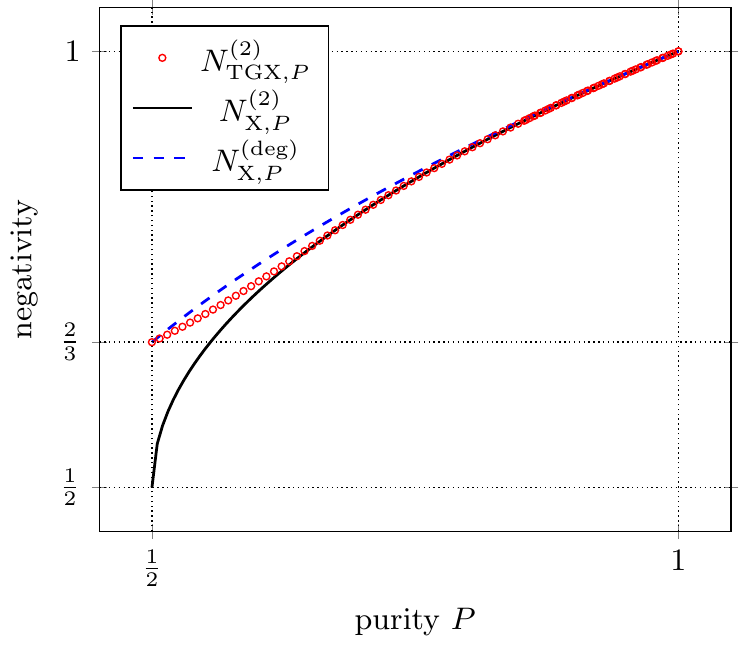}
\caption{(Color online) Graphical demonstration that there exist $2\times 3$ rank-2 nonX states of purity $P$ that reach higher negativities than the rank-2 X-MEMS of purity $P$, that is $N_{\text{TGX},P}^{(2)}\geq N_{{\rm X},P}^{(2)}$. The states reaching the negativity $N_{\text{TGX},P}^{(2)}$ are parametrized by Eq.~(\ref{eq:rhoTGX2}) and the parameter values were numerically obtained, whereas the X states reaching negativities $N_{{\rm X},P}^{(2)}$ and $N_{{\rm X},P}^{(\text{deg})}$ were analytically constructed in Theorems~\ref{thm:rank2} and \ref{thm:deg}, respectively. Note that $N_{{\rm X},P}^{(\text{deg})}$ corresponds to \textit{higher-than}-rank-2 states and is only shown as a reference. This semi-analytical plot further supports the numerically generated plot of Fig. 9 of Ref.~\cite{13Hedemann}, which also showed evidence that $N_{\text{TGX},P}^{(2)}$ may be the upper bound of all rank-2 general $2\times 3$ states.}
\label{fig:TGXrank2}
\end{figure}

Figure~\ref{fig:TGXrank2} shows that $N_{\text{TGX},P}^{(2)}$ is generally greater than $N_{{\rm X},P}^{(2)}$, hence the general $2\times 3$ states output by our numerical routine cannot be unitarily transformed into X states of the same negativity. We remark that this conclusion neither depends on assuming that our numerical maximizations converged to global optima, nor that Eq.~(\ref{eq:rhoTGX2}) gives a valid parametrization for $2\times 3$ rank-2 TGX states. A refutation of any one of these hypothesis would only mean that $2\times 3$ rank-2 states of purity $P$ with even greater negativity could be found.

It is interesting to ask whether rank-$r$ $2\times 3$ states with purity-parametric negativity greater than $N_{{\rm X},P}^{(r)}$ also occur for $r>2$. In what follows, we provide some numerical evidence that this may not be the case already for $r=3$. As in the $r=2$ case, we start with Hedemann's candidate parametrization for $2\times 3$ rank-3 TGX states~\cite{13Hedemann},
\begin{equation}\label{eq:rhoTGX3}
\bm{\varrho}^{(3)}_{\text{TGX}} \!=\!\frac{1}{2}\!\left(\! {\begin{array}{cccccc}
   2 p_1 c_{\theta_1}^2 &\!\!  \cdot  &\!\!  \cdot  &\!\!  \cdot  &\!\!  \cdot  &\!\! p_1 s_{2\theta _1}\\
    \cdot  &\!\! 2 p_2 c_{\theta _2 }^2  &\!\! \cdot &\!\! p_2 s_{2\theta_2}  &\!\!  \cdot  &\!\! \cdot   \\
    \cdot  &\!\!  \cdot  &\!\!  2 p_3 c_{\theta _3 }^2  &\!\!  \cdot  &\!\!  p_3 s_{2\theta_3}  &\!\!  \cdot   \\
    \cdot  &\!\! p_2 s_{2\theta_2} &\!\!  \cdot  &\!\! 2 p_2 s_{\theta _2}^2 &\!\!  \cdot  &\!\! \cdot \\
    \cdot  &\!\!  \cdot  &\!\!  p_3 s_{2\theta_3}  &\!\!  \cdot  &\!\!  2 p_3 s_{\theta _3 }^2  &\!\!  \cdot   \\
   p_1 s_{2\theta _1} &\!\!  \cdot  &\!\!  \cdot  &\!\!  \cdot  &\!\!  \cdot  &\!\! 2 p_1 s_{\theta _1 }^2 \\
\end{array}}\!\! \right)\!,
\end{equation}
whose nonzero eigenvalues are parametrically given by probabilities $p_1$, $p_2$, and $p_3$ where $p_{1},p_{2},p_{3}> 0$ and $p_1+p_2+p_3=1$, whereas its purity is $P = p_1^2+p_2^2+p_3^2$, and its negativity $\mathcal{N}_{\text{TGX}}^{(3)}=\mathcal{N}_{\text{TGX}}^{(3)}(\theta_1,\theta_2,\theta_3,p_1,p_2,p_3)$ is
\begin{equation}
  \mathcal{N}_{\text{TGX}}^{(3)}=  \sum_{\ell=1}^3 |\sigma_\ell|-\sigma_\ell\,,
\end{equation}
with $\{\sigma_\ell\}_{\ell=1}^3$ denoting the three possibly negative eigenvalues of the partial transpose of $\bm{\varrho}^{(3)}_{\text{TGX}}$, explicitly,
\begin{align}
\sigma_{4-k}=&\frac{1}{2}\left(p_i\sin^2\theta_i+p_j\cos^2\theta_j\right)\nonumber\\
&-\frac{1}{2}\sqrt{p_k^2\sin^2(2\theta_k)+\left(p_i\sin^2\theta_i-p_j\cos^2\theta_j\right)^2},
\end{align}
where $(i,j,k)$ must be taken as a cyclic permutation of $(1,2,3)$. In this case, the specification of $P$ does not fully specify $p_1$, $p_2$, and $p_3$, for which reason the maximization of $\mathcal{N}_{\text{TGX}}^{(3)}$ must be taken with respect to the variables $\theta_{1},\theta_{2},\theta_{3}$ \emph{and} $p_{1},p_{2},p_{3}$, with constraints $p_{1},p_{2},p_{3}>0$, $p_1+p_2+p_3=1$, and $p_1^2+p_2^2+p_3^2=P$. By numerically implementing this optimization problem in MATLAB\footnote{Here, we employ the MATLAB function \texttt{fmincon}, running the active-set algorithm with a random initial guess for $P=1/3$ and, for $P>1/3$, with initial guesses equal to the converged variable values obtained in the previous optimization (for the previous value of $P$). The termination tolerances for the parameter values (\texttt{TolX}) and objective function value (\texttt{TolFun}) were set to $10^{-6}$, whereas the tolerance for the maximal constraint violation (\texttt{TolCon}) was set to $10^{-15}$.} 
\begin{figure}[h]
\centering
\includegraphics[width=\columnwidth]{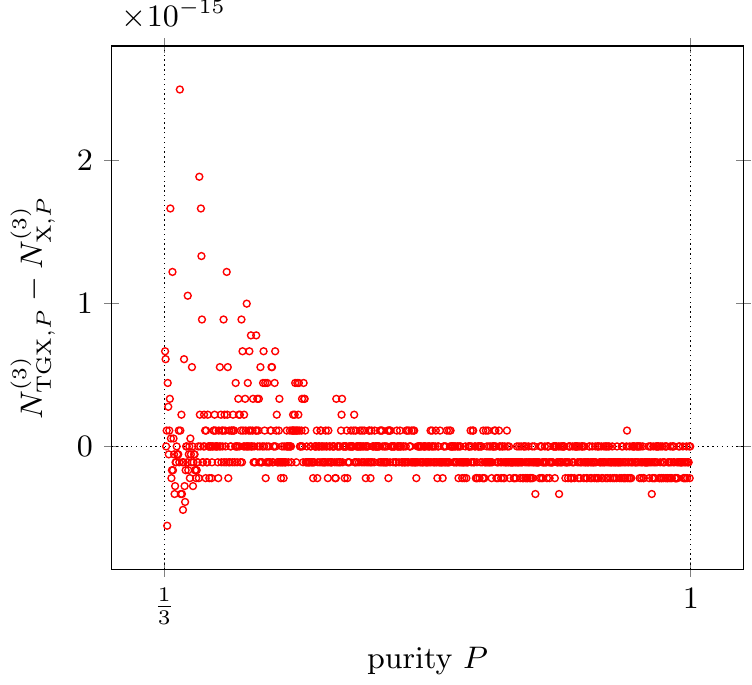}
\caption{(Color online) Numerical evidence that $N_{\text{TGX},P}^{(3)}=N_{{\rm X},P}^{(3)}$ for $P\in[\frac{1}{3},1[$, suggesting that the $2\times 3$ rank-3 X-MEMS wrt purity achieve the highest value of negativity per purity attainable by $2\times 3$ rank-3 TGX states.}\label{fig:TGXrank3}
\end{figure}
for 1000 values of $P$ uniformly sampled in the interval $[1/3,1[$, we find that the maximal values of $\mathcal{N}_{\text{TGX}}^{(3)}$ (henceforth denoted by $N_{\text{TGX},P}^{(3)}$), match $N_{{\rm X},P}^{(3)}$ to the order of numerical precision $10^{-15}$, as shown in Fig.~\ref{fig:TGXrank3}.


Confirmation that $N_{\text{TGX},P}^{(3)}=N_{{\rm X},P}^{(3)}$ would require rigorous substantiation of the adopted premises, namely that Eq.~(\ref{eq:rhoTGX3}) does parametrize a family of $2\times 3$ rank-3 TGX states and that our numerical procedure has not output local (nonglobal) optima. Moreover, it should be stressed that even if these premises were confirmed, we could still not rely solely on $N_{\text{TGX},P}^{(3)}=N_{{\rm X},P}^{(3)}$ to claim that any $2\times 3$ rank-3 state can be unitarily mapped into a $2\times 3$ rank-3 X state of the same negativity. That is because, for $r>2$, the condition of having two states with same purity and rank does not imply that these states are unitarily related (cf. Ref.~\cite[App.~C]{14Mendonca79}).

\section{$\mathbf{2\times 3}$ MEMS wrt purity}\label{sec:MEMS}

In this section we drop the X form and all the spectral constraints admitted in Sec.~\ref{sec:spconstrainedXMEMS} to search for MEMS wrt purity amongst all $2\times 3$ states of purity $P$. From a mathematical standpoint, that amounts to finding optimal solutions for the (nonconvex) problem
\begin{align}
N_P\coloneq& \max_{\bm{\rho}_P} \|\bm{\rho}_P^\Gamma\|_{\text{tr}}-1 \nonumber\\
&\mbox{such that}\quad \bm{\rho}_P\geq 0\,,\quad \tr\bm{\rho}_P=1\,,\quad \tr\bm{\rho}_P^2\leq P\,,\label{eq:optgenstep1}
\end{align}
where the negativity of a legitimate $2\times 3$ state is maximized under the sole constraint that $P$ is fixed\footnote{Once again, we rely on the ``Convex Optimization Maximum Principle'' to replace the purity constraint $\tr \bm{\rho}_P^2=P$ with $\tr\bm{\rho}_P^2\leq P$. As noted before, this has no impact whatsoever on the solution of the optimization problem since the maximum of a convex function over a convex set necessarily occurs on the boundary of that set, in which case $\tr \bm{\rho}_P^2=P$ will be satisfied even under the weaker requirement $\tr\bm{\rho}_P^2\leq P$.} within $]\frac{1}{5},1[$. Notice that apart from nonnegativity, normalization, and degree of mixedness, no further constraints (e.g, sparse structure, rank-deficiency, or degeneracy) have been imposed on $\bm{\rho}_P$. As a result, $N_P$ is the highest purity-parametric negativity amongst \emph{all} $2\times 3$ states.

In order to deal with problem~(\ref{eq:optgenstep1}), we first invoke a key variational characterization of the trace norm of an arbitrary Hermitian operator $\bm{A}$~\cite[Lemma 4]{07Rastegin9533},
\begin{equation}
\|\bm{A}\|_{\text{tr}}=-\tr\bm{A}+2\max_{0\leq\bm{\Pi}\leq\bm{I}}\tr[\bm{\Pi A}]\,,
\end{equation}
which, applied to problem~(\ref{eq:optgenstep1}) (along with the fact that $\tr\bm{\rho}_P^\Gamma=1$), leads to
\begin{align}
N_P=&-2+2\max_{\bm{\rho}_P,\bm{\Pi}} \tr[\bm{\Pi}\bm{\rho}_P^\Gamma]\nonumber\\
&\mbox{such that}\quad \bm{\Pi}\geq 0\,,\quad \bm{I}-\bm{\Pi}\geq 0\,,\nonumber\\
& \bm{\rho}_P\geq 0\,,\quad \tr\bm{\rho}_P=1\,,\quad \tr\bm{\rho}_P^2\leq P\,.\label{eq:optprblm3}
\end{align}

In spite of the bilinearity of the objective function $\tr[\bm{\Pi}\bm{\rho}_P^\Gamma]$, problem~(\ref{eq:optprblm3}) reduces to a convex optimization (an SDP) if either $\bm{\Pi}$ or $\bm{\rho}_P$ are held fixed. In cases like this, it is customary to approach the nonconvex optimization problem with an \emph{alternate convex search} (ACS) algorithm~\cite{76Wendell643,94Ghaoui2678,07Gorski373,09Kosut443} which, starting from some initial guess for one of the variables, solves a convex optimization for the other variable. Then, fixing the latter at the solution just obtained solves another convex optimization for the former. Such optimization rounds are iterated until convergence is attained.

Figure~\ref{fig:ACSevolution} illustrates a typical run of the ACS scheme for problem~(\ref{eq:optprblm3}), with initial guess $\bm{\rho}_{P,0}$ taken as a set of randomly generated full-rank $2\times 3$ states of purity $P$ uniformly sampled over $]\frac{1}{5},1[$. The plots provide a partial view of the evolution of the figure-of-merit, 
\begin{equation}
N_{P,n}\coloneq -2+2\tr[\bm{\Pi}_n\bm{\rho}_{P,n}^\Gamma]\,,
\end{equation}  
for a few values of $n$, which labels the outputs of the $n$th optimization round. The left scale of each plot applies to the output negativity at the current optimization round (plus signs), whereas the right scale applies to the difference between the output negativities in the current and previous rounds (empty circles).

\begin{figure*}[t]
\centering
\includegraphics[width=\textwidth]{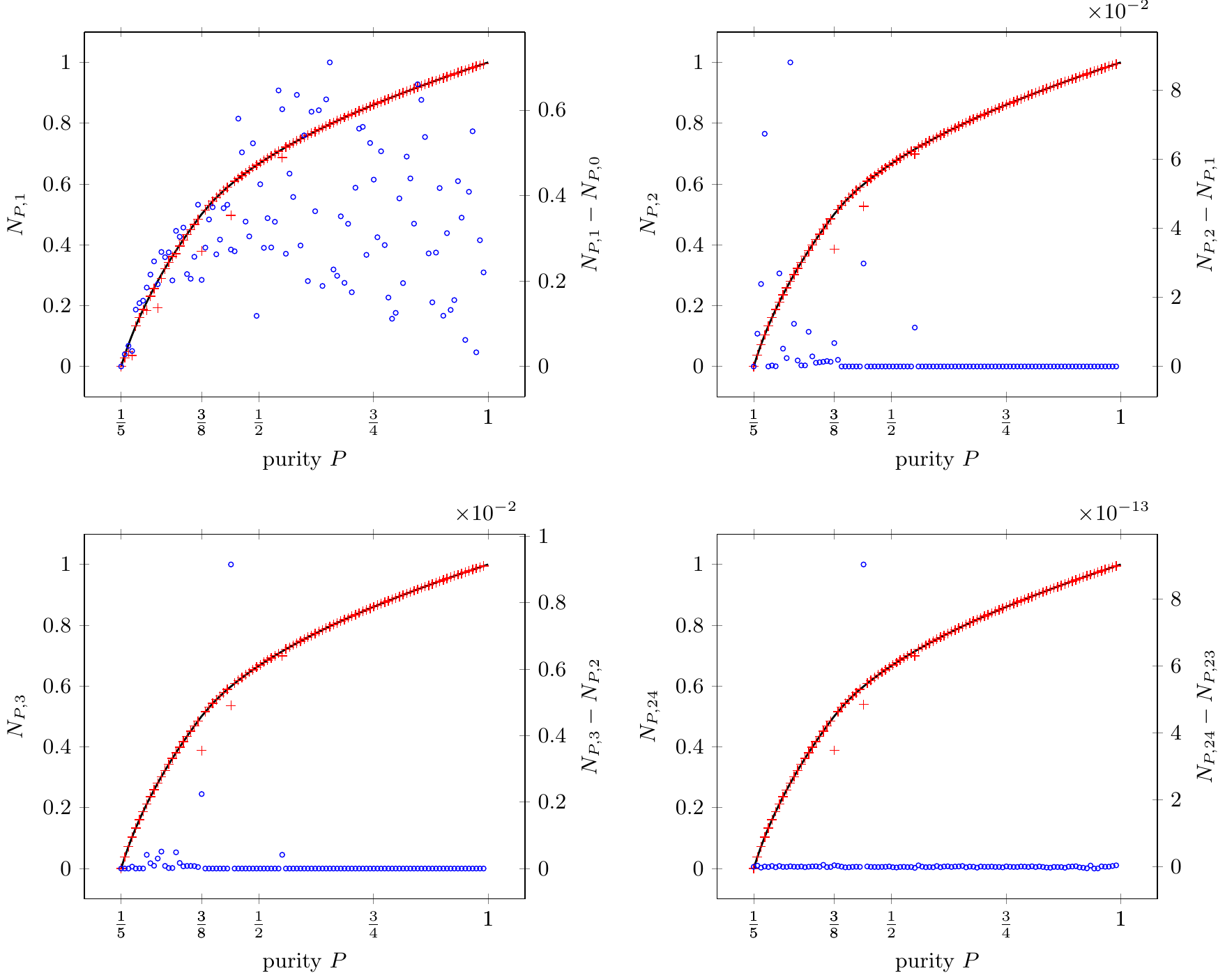}
  \caption{(Color online) A typical evolution of negativity during the ACS scheme. Starting from a random set of 100 full-rank $2\times 3$ states of purities $P$ uniformly sampled over the interval $]\frac{1}{5},1[$, SDPs are solved to obtain their negativities by optimally choosing $\Pi_0$. After this initialization round ($n=0$), subsequent rounds iteratively optimize $\rho_{P,n}$ (for $\Pi_n$ set to $\Pi_{n-1}$) and $\Pi_n$ (for $\rho_{P,n}$ set to $\rho_{P,n-1}$, the state obtained from the previous iteration) until no more significant negativity increment is observed between consecutive rounds (convergence). Plots (a), (b), (c), and (d) show the negativities (+, with its scale on the left vertical axis) and their difference wrt the previous round ($\circ$, with its scale on the right vertical axis) output at the $n=1$, $n=2$, $n=3$, and $n=24$ optimization rounds, respectively. For the great majority of considered purities, the converged negativity values match $N_{{\rm X},P}^{(\text{deg})}$, represented in the plots by the solid line.}\label{fig:ACSevolution}
\end{figure*}

For the numerical experiment depicted in Fig.~\ref{fig:ACSevolution}, a number of $24$ optimization rounds were necessary to fulfill the adopted convergence (stop) criterion
\begin{equation}\label{eq:stopcriterion}
N_{P,n}-N_{P,n-1}<10^{-12}
\end{equation}
for \emph{all} considered values of $P\in]\frac{1}{5},1[$. Surprisingly, for the great majority of purity values, the converged negativity values match $N_{{\rm X},P}^{(\text{deg})}$ to numerical precision [in order to provide some visual reference, Eq.~(\ref{eq:NXPdeg}) is plotted as a solid line in each graph], the only exceptions being a few scattered points that converged to \emph{smaller} negativities. 

The occurrence of these ``subconverged points'' is a direct consequence of the fact that ACS is not guaranteed to converge to a global optimum~\cite{76Wendell643}. As a matter of fact, one cannot even be sure that the massive convergence to $N_{{\rm X},P}^{(\text{deg})}$ represents global optimum convergence. In this framework, the best one can do is to repeat as many runs as possible of the ACS scheme (each of which starting form a different initial guess) and hope that global optimality is attained in at least one of them, for at least one value of $P$.

In our numerical experiments we have performed over a thousand runs of the ACS scheme and have never observed a single convergence to a value greater than $N_{{\rm X},P}^{(\text{deg})}$. This observation strongly \emph{suggests} that $N_{{\rm X},P}^{(\text{deg})}$ is indeed the global optimum, or, what amounts to be the same, that $\bm{\varrho}_{{\rm X},P}^{(\text{deg})}$ [cf. Eqs.~(\ref{eq:rhoXPdeg}) and (\ref{eq:lambdasdeg})] represents a family of $2\times 3$ MEMS wrt to purity. If rigorously confirmed, this observation would imply that X states are rich enough to subsume MEMS wrt purity (without rank constraints) for qubit-qutrit systems, as is well-known to be the case for two-qubit systems~\cite{01Munro30302,03Wei22110} (even if rank constraints are applied for the two qubits). Unfortunately though, owing to the convergence properties of the ACS scheme, this can only be stated as a conjecture so far.

\section{Concluding Remarks}\label{sec:conclusion}

By formulating and analytically solving nonlinear optimization problems, we have characterized families of $2\times 3$ X states that reach maximal entanglement negativity either for a given spectrum, or for a given purity \emph{and} one of the following extra constraints: rank $2$, rank $3$, or three-fold degeneracy of the smallest eigenvalue. In so doing, we have refuted a current candidate of $2 \times 3$ MEMS wrt purity~\cite{13Hedemann} and replaced it with a slightly more entangled family of X states, whose overall optimality is supported by compelling numerical evidence. 

An intriguing byproduct of our X-MEMS constructions was the observation that $2\times 3$ rank-2 X-MEMS of purity $P$ can have their negativity exceeded by more general (nonX) $2\times 3$ rank-2 states of the same purity, implying that not every $2\times 3$ state can be \emph{unitarily} mapped into an X state of same entanglement. This simple observation reveals the set of $2\times 3$ X states to be a ``less universal'' set than that of $2 \times 2$ X states, since any $2\times 2$ state can be mapped into an X state via an entanglement-preserving unitary transformation; a property that has been recently entitled ``entanglement universality of two-qubit X states''~\cite{13Hedemann,14Mendonca79}.

This observed lack of entanglement universality of $2\times 3$ X states was a conjecture in Ref.~\cite{13Hedemann}, where alternative matrix forms were postulated with the intent of establishing a universal family of $2\times 3$ states --- the so-called TGX states. Although our results suffice to confirm the above-mentioned universality breach, we can neither confirm nor disprove the entanglement universality of the $2\times 3$ TGX candidate states proposed by Hedemann. However, as a favorable note, we emphasize that the breach of entanglement universality in $2\times 3$ rank-2 X states was revealed by the highest negativities reached by Hedemann's $2\times 3$ rank-2 TGX candidate states. In addition, it should also be noted that our numerical analysis for rank-3 states reinforces the preexisting thesis~\cite{13Hedemann} that the candidate rank-3 TGX family does just as well as the rank-3 X family as far as the maximization of negativity per purity is concerned. However, since we cannot be sure that rank-3 X-MEMS are actually rank-3 MEMS, nothing can be concluded on the validity of the proposed matrix form for the candidate rank-3 TGX states.

We conclude by discussing some possible directions for future work. First, it is conceivable that a weaker entanglement universality property holds for $2\times 3$ X states, namely, \emph{it is always possible to nonunitarily map an arbitrary $2\times 3$ state into a $2\times 3$ X state of same negativity and purity\emph}. Notice that such a hypothesis does not contradict our results, since the nonunitarity of the mapping would allow for rank changes between the input and output states. Furthermore, the offered numerical evidence that $2\times 3$ MEMS wrt purity can always be made X shaped (as long as no rank constraints are enforced), reinforces our confidence in this thesis.

Finally, it would be very fruitful to replace the ACS scheme with a global optimization strategy for the bilinear problem~(\ref{eq:optprblm3}). In that case, a resulting match between the converged value and $N_{{\rm X},P}^{(\text{deg})}$ would leave no doubt that $\bm{\varrho}_{{\rm X},P}^{(\text{deg})}$ actually parametrizes a $2\times 3$ MEMS family with respect to purity. A number of global optimization algorithms for bilinear programs, based on branch-and-bound~\cite{60Land497,63Little972} and Benders decomposition~\cite{62Benders238,72Geoffrion237} strategies, have appeared in the applied mathematics literature~\cite{94Goh2009,00VanAntwerp363,90Floudas1397,90Visweswaran1419,93Floudas178,97Beran,00Tuan561} and call for further investigation as to whether (and how) they could be applied to the problem at hand. Another promising direction of investigation involves the relaxation theory of nonconvex problems~\cite{01Lasserre796,03Parrilo293,03Parrilo83}, where hierarchies of SDP relaxations are devised and, at each step of the hierarchy, a better approximation of the global solution is attained. Such a method has already been employed in several problems in the context of quantum information theory~\cite{04Doherty022308,04Eisert062317,07Navascues010401,07Liang042103}.

In closing, our results reveal many interesting features of entanglement in $2\times 3$ systems, as well as new directions to pursue. It is our hope that this identification of explicit forms of $2\times 3$ MEMS will prompt further research on the role of entanglement in quantum information tasks involving real-world (i.e., mixed) low-dimensional quantum systems.

\begin{acknowledgments}
PEMFM thanks Alessandro Firmiano for discussions and, in particular, for collaboration in the proof of Proposition~\ref{proposition1}.
\end{acknowledgments}

\appendix
\section{Negative eigenvalues of partial-transposed of $\mathbf{2\times 3}$ X states}\label{app:onenegeigenvalue}
In Ref.~\cite{13Rana054301}, Rana demonstrated that partial transposition of an arbitrary $m \times n$ state has at most $(m-1)(n-1)$ negative eigenvalues, which implies an upper bound of $2$ negative eigenvalues for $2\times 3$ systems. Here, we show that this upper bound drops to $1$ for \emph{X states} in $2\times 3$ systems.

Let $\bm{\rho}_{\rm X}^{\Gamma}$ be the partial transpose of a $2\times 3$ X state, with $\bm{\rho}_{\rm X}$ parametrized as in Eq.~(\ref{eq:qubitqutritXstate}). We have already seen [cf. Eq.~(\ref{eq:lambdaprime})] that the six eigenvalues of $\bm{\rho}_{\rm X}^{\Gamma}$ are
\begin{equation}\label{eq:lambdaprime_repeat}
\lambda_k^{\prime\pm}=\frac{a_k+b_k}{2}\pm\sqrt{r_{4-k}^2+\left(\frac{a_k-b_k}{2}\right)^2},
\end{equation}
for $k=1,2,3$. Since $a_k,b_k\geq 0$, then that $\lambda_k^{\prime +} \geq 0$, and since $\lambda_2^{\prime -}$ matches one of the eigenvalues of $\bm{\rho}_{\rm X}$ [cf. Eq.~(\ref{eq:lambdas})], then it is nonnegative as well. Therefore, the only possibly negative eigenvalues are $\lambda_1^{\prime -}$ and $\lambda_3^{\prime -}$.

Now, suppose $\lambda_1^{\prime -}$ and $\lambda_3^{\prime -}$ are both negative. From Eq.~(\ref{eq:lambdaprime_repeat}), it is straightforward to show that $\lambda_k^{\prime -}<0$ is equivalent to $r_{4-k}>\sqrt{a_kb_k}$, in which case our hypothesis imposes the simultaneous fulfillment of the inequalities $r_3>\sqrt{a_1b_1}$ and $r_1>\sqrt{a_3 b_3}$. However, recall that the positive-semidefiniteness of $\bm{\rho}_{\rm X}$ requires that $r_3\leq\sqrt{a_3 b_3}$ and $r_1\leq\sqrt{a_1b_1}$ [cf. Eq.~(\ref{eq:qualifiers})] in such a way that the following inequalities must be simultaneously satisfied;
\begin{equation}
\sqrt{a_3b_3}\geq r_3>\sqrt{a_1b_1}\quad\mbox{and}\quad \sqrt{a_1b_1}\geq r_1>\sqrt{a_3 b_3}.
\end{equation}
These inequalities impose conflicting orderings between $\sqrt{a_1 b_1}$ and $\sqrt{a_3 b_3}$ and hence cannot be simultaneously satisfied. Therefore, either $\lambda_1^{\prime -}<0$ or $\lambda_3^{\prime -}<0$, that is, entangled $2\times 3$ X states have exactly one negative eigenvalue in their partial transposes.

\section{Optimal solution for problem~(\ref{eq:proboptdist})}\label{app:discopt}
Here, we prove the following proposition that can be directly applied to establish Eq.~(\ref{eq:onlyoptsol}) as the optimal solution for problem~(\ref{eq:proboptdist}) (with added constraints $\lambda_k^+\geq\lambda_k^-$ for $k\in\{1,2,3\}$).
\begin{proposition}\label{proposition1}
For any given set of real numbers $\{\lambda_\ell\}_{\ell=1}^6$ such that  $1\geq\lambda_1\geq\lambda_2\geq\lambda_3\geq\lambda_4\geq\lambda_5\geq\lambda_6\geq 0$, the largest value of
\begin{equation}
\mathcal{S}_{i,j,k,l}\coloneq -(\lambda_i+\lambda_j)+\sqrt{(\lambda_i-\lambda_j)^2+(\lambda_k-\lambda_l)^2},
\end{equation}
taken amongst all possible sequences $(i,j,k,l)$ that satisfy $\{i,j,k,l,m,n\}=\{1,2,3,4,5,6\}$, occurs for
\begin{equation}
\{i,j\}=\{4,6\}\quad\mbox{and}\quad \{k,l\}=\{1,5\}.
\end{equation}
\end{proposition}
Before proceeding with a proof for this proposition, let us state and prove some useful inequalities.
\begin{lemma}\label{lemma1}
The inequalities
\begin{align}
a+\sqrt{(a+b)^2+(b+c)^2}&\geq\sqrt{b^2+(a+b+c)^2}\,,\label{eq:ineqabc}\\
\sqrt{(b+a)^2+(c+a)^2}&\geq \sqrt{b^2+c^2}+a\,,\label{eq:ineqabc2}\\
\sqrt{b^2+c^2}+a&\geq\sqrt{(b+a)^2+c^2}\,, \label{eq:ineqabc3}
\end{align}
hold for all nonnegative real numbers $a$, $b$, and $c$.
\end{lemma}
\begin{proof}(Lemma~\ref{lemma1})
We offer separate proofs for each inequality;
\begin{itemize}[leftmargin=*]
\item \emph{Inequality~(\ref{eq:ineqabc}).} Since both sides of~(\ref{eq:ineqabc}) are nonnegative, an equivalent inequality can be obtained by squaring~(\ref{eq:ineqabc}). After some straightforward manipulation we arrive at
\begin{equation}\label{eq:ineqabcsq1}
a(a-2c)\geq -2a\sqrt{(a+b)^2+(b+c)^2},
\end{equation} 
which is fulfilled if $a=0$ or $a\geq 2c$. Now, there only remains to prove~(\ref{eq:ineqabcsq1}) for $0<a<2c$. In this case, inequality~(\ref{eq:ineqabcsq1}) can be written as
\begin{equation}\label{eq:ineqabcsq2}
|a-2c|\leq 2\sqrt{(a+b)^2+(b+c)^2},
\end{equation} 
which can be squared to yield the equivalent inequality
\begin{equation}\label{eq:ineqabcsq3}
3a^2+8b(b+c)+4a(2b+c)\geq 0\,,
\end{equation}
which holds due to the nonnegativity of each term on its lhs.
\item \emph{Inequality~(\ref{eq:ineqabc2}).} As before, both sides of~(\ref{eq:ineqabc2}) are nonnegative and an equivalent inequality can be obtained by squaring them as
\begin{equation}
a[a+2(b+c)]\geq a(2\sqrt{b^2+c^2}).
\end{equation}
This inequality is saturated for $a=0$ and, for $a>0$, it takes the form
\begin{equation}
a+2(b+c)\geq 2\sqrt{b^2+c^2},
\end{equation}
which holds since $a+2(b+c)>2(b+c)\geq 2\sqrt{b^2+c^2}$ (triangle inequality).
\item \emph{Inequality~(\ref{eq:ineqabc3}).} Once again, owing to the nonnegativity of both sides of~(\ref{eq:ineqabc3}), an equivalent inequality can be obtained by squaring~(\ref{eq:ineqabc3}) as
\begin{equation}
2a\sqrt{b^2+c^2}\geq 2 a b,
\end{equation}
which is saturated when $a=0$ and, for $a>0$, is equivalent to the true statement $\sqrt{b^2+c^2}\geq b$.
\end{itemize}
\end{proof}

We now prove Proposition~\ref{proposition1}. Firstly, note that $\mathcal{S}_{i,j,k,l}$ is invariant under the exchanges of $i\leftrightarrow j$ and $k\leftrightarrow l$, in such a way that it suffices to search for an optimal sequence $(i,j,k,l)$ with $i<j$ and $k<l$. Secondly, for fixed values of $i$ and $j$, the values of $k$ and $l$ (such that $k<l$) that maximize $S_{i,j,k,l}$ are
\begin{align}\label{eq:notationissueeq}
k&=\min[\{1,2,3,4,5,6\}\setminus\{i,j\}]\\
l&=\max[\{1,2,3,4,5,6\}\setminus\{i,j\}]
\end{align}
since such a choice maximizes $(\lambda_k-\lambda_l)^2$ for any fixed $i$ and $j$, where the notation here means that the values of $i$ and $j$ are omitted from the full set. With these consideration in mind, we are left with fifteen sequences of the form $(i,j,k,l)$ that can be considered candidates to the maximization of $S_{i,j,k,l}$, namely:  1. $(1,2,3,6)$; 2. $(1,3,2,6)$; 3. $(1,4,2,6)$; 4. $(1,5,2,6)$; 5. $(1,6,2,5)$; 6. $(2,3,1,6)$; 7. $(2,4,1,6)$; 8. $(2,5,1,6)$; 9. $(2,6,1,5)$; 10. $(3,4,1,6)$; 11. $(3,5,1,6)$; 12. $(3,6,1,5)$; 13. $(4,5,1,6)$; 14. $(4,6,1,5)$; 15. $(5,6,1,4)$. 

In order to establish sequence 14 as the maximizer of $S_{i,j,k,l}$, we take two main steps. First, we reduce the number of candidate sequences from fifteen to eight, by demonstrating that (i) $\mathcal{S}_{2,6,1,5}\geq \mathcal{S}_{2,j,1,6}$ for $j\in\{3,4,5\}$ (so that sequences 6, 7, and 8 can be disregarded); (ii) $\mathcal{S}_{3,6,1,5}\geq \mathcal{S}_{3,j,1,6}$ for $j\in\{4,5\}$ (so that sequences 10 and 11 can be disregarded); (iii) $\mathcal{S}_{4,6,1,5}\geq\mathcal{S}_{4,5,1,6}$ (so that sequence 13 can be disregarded); (iv) $\mathcal{S}_{4,6,1,5}\geq\mathcal{S}_{5,6,1,4}$ (so that sequence 15 can be disregarded). At this point, the only remaining candidate sequences are: 1. $(1,2,3,6)$; 2. $(1,3,2,6)$; 3. $(1,4,2,6)$; 4. $(1,5,2,6)$; 5. $(1,6,2,5)$; 9. $(2,6,1,5)$; 12. $(3,6,1,5)$; 14. $(4,6,1,5)$. 

Next, we discard all remaining sequences but the last by demonstrating that (v) $\mathcal{S}_{4,6,1,5}\geq\mathcal{S}_{1,j,k,6}$ for ($j=2$ and $k=3$) or ($j\in\{3,4,5\}$ and $k=2$) (so that sequences 1, 2, 3, and 4 can be disregarded); (vi) $\mathcal{S}_{4,6,1,5}\geq\mathcal{S}_{1,6,2,5}$ (so that sequences 5 can be disregarded) and (vii) $\mathcal{S}_{4,6,1,5}\geq\mathcal{S}_{i,6,1,5}$ for $i\in\{2,3\}$ (so that sequences 9 and 12 can be disregarded). 

In order to prove the aforementioned inequalities, it is convenient to reexpress the $\lambda$-parameters in terms of the $\delta$-distances between them, as implicitly defined in Fig.~\ref{fig:defdeltai}.
\begin{figure}[h]
\centering
\includegraphics[width=\columnwidth]{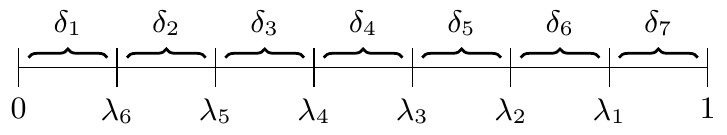}
\caption{Definition of $\delta_i$ for every $i\in\{1,2,3,4,5,6,7\}$. Note that the distances $\delta_i$ are not generally equal.}\label{fig:defdeltai}
\end{figure}
\noindent From that figure, we note that $\sum_{i=1}^7\delta_i=1$ and that $\delta_i\geq 0$ for every $i\in\{1,2,3,4,5,6,7\}$. Although Fig.~\ref{fig:defdeltai} shows all $\delta_i$ of the same length, this is not generally the case. In fact, different spectra will yield all possible orderings of $\delta_i$ (as opposed to $\lambda_i$, which are always, by definition, in descending order for ascending $i$). Moreover, any element of $\{\lambda_i\}_{i=1}^6$ can be written in terms of $\delta$-distances as
\begin{equation}
\lambda_i=\sum_{p=1}^{7-i}\delta_p,
\end{equation}
in which case, the expression of $\mathcal{S}_{i,j,k,l}$ with $i<j$ and $k<l$ takes the following form
\onecolumngrid
\begin{equation}\label{eq:Sijkldelta}
\mathcal{S}_{i,j,k,l}=-\sum_{p=7-(j-1)}^{7-i}\delta_p-2\sum_{q=1}^{7-j}\delta_q+\sqrt{\left(\sum_{p=7-(j-1)}^{7-i}\delta_p\right)^2+\left(\sum_{r=7-(l-1)}^{7-k}\delta_r\right)^2}\,.
\end{equation}

With these provisions, we are ready to prove the aforementioned inequalities
\begin{itemize}[leftmargin=*]
\item[i.]
Proof that $\mathcal{S}_{2,6,1,5}\geq \mathcal{S}_{2,j,1,6}$ for $j\in\{3,4,5\}$. In terms of the $\delta$-distances, this inequality takes the form [cf. Eq.~(\ref{eq:Sijkldelta})]
\begin{equation}
-\sum_{p=2}^{5}\delta_p-2\delta_1+\sqrt{\left(\sum_{p=2}^{5}\delta_p\right)^2+\left(\sum_{r=3}^{6}\delta_r\right)^2}
\geq
-\sum_{p=7-(j-1)}^{5}\delta_p-2\sum_{q=1}^{7-j}\delta_q+\sqrt{\left(\sum_{p=7-(j-1)}^{5}\delta_p\right)^2+\left(\sum_{r=2}^{6}\delta_r\right)^2},
\end{equation}
which, letting $\Delta_{345}\coloneq \delta_3+\delta_4+\delta_5$, can be trivially rearranged into
\begin{equation}\label{eq:ineqcaseiori}
\sum_{p=2}^{7-j}\delta_p+\sqrt{\left(\delta_2+\Delta_{345}\right)^2+\left(\Delta_{345}+\delta_6\right)^2}
\geq
\sqrt{\left(\sum_{p=7-(j-1)}^{5}\delta_p\right)^2+\left(\delta_2+\Delta_{345}+\delta_6\right)^2}.
\end{equation}
Now, note that for all possible values of $j\in\{3,4,5\}$, the case $j=5$ leads to the strongest inequality, as it minimizes the sum on the lhs and maximizes the sum on the rhs. So, if Eq.~(\ref{eq:ineqcaseiori}) holds for $j=5$, then it will necessarily hold for $j=3$ and $j=4$ as well. For this reason, it suffices to consider Eq.~(\ref{eq:ineqcaseiori}) with $j=5$, namely
\begin{equation}\label{eq:ineqcaseiori2}
\delta_2+\sqrt{\left(\delta_2+\Delta_{345}\right)^2+\left(\Delta_{345}+\delta_6\right)^2}
\geq
\sqrt{\Delta_{345}^2+\left(\delta_2+\Delta_{345}+\delta_6\right)^2}\,.
\end{equation}
The validity of Eq.~(\ref{eq:ineqcaseiori2}) is certified by Eq.~(\ref{eq:ineqabc}) (Lemma~\ref{lemma1}) with $a=\delta_2$, $b=\Delta_{345}$, and $c=\delta_6$. 

\item[ii.]
Proof that $\mathcal{S}_{3,6,1,5}\geq \mathcal{S}_{3,j,1,6}$ for $j\in\{4,5\}$. In terms of the $\delta$-distances, this inequality takes the form [cf. Eq.~(\ref{eq:Sijkldelta})]
\begin{equation}
-\sum_{p=2}^{4}\delta_p-2\delta_1+\sqrt{\left(\sum_{p=2}^{4}\delta_p\right)^2+\left(\sum_{r=3}^{6}\delta_r\right)^2}\geq -\sum_{p=7-(j-1)}^{4}\delta_p-2\sum_{q=1}^{7-j}\delta_q+\sqrt{\left(\sum_{p=7-(j-1)}^{4}\delta_p\right)^2+\left(\sum_{r=2}^{6}\delta_r\right)^2},
\end{equation}
which, letting $\Delta_{34}\coloneq \delta_3+\delta_4$ and $\Delta_{56}\coloneq\delta_5+\delta_6$, can be trivially rearranged into
\begin{equation}\label{eq:ineqcaseiiori}
\sum_{p=2}^{7-j}\delta_p+\sqrt{\left(\delta_2+\Delta_{34}\right)^2+\left(\Delta_{34}+\Delta_{56}\right)^2}
\geq
\sqrt{\left(\sum_{p=7-(j-1)}^{4}\delta_p\right)^2+\left(\delta_2+\Delta_{34}+\Delta_{56}\right)^2}\,.
\end{equation}
As in the previous case, although $j$ can take any value in the set $\{4,5\}$, it suffices to consider $j=5$ as this choice yields the strongest inequality,
\begin{equation}\label{eq:ineqcaseiiorij5}
\delta_2+\sqrt{\left(\delta_2+\Delta_{34}\right)^2+\left(\Delta_{34}+\Delta_{56}\right)^2}
\geq
\sqrt{\Delta_{34}^2+\left(\delta_2+\Delta_{34}+\Delta_{56}\right)^2}\,.
\end{equation}
The validity of this inequality is certified by Eq.~(\ref{eq:ineqabc}) (Lemma~\ref{lemma1}) with $a=\delta_2$, $b=\Delta_{34}$, and $c=\Delta_{56}$.

\item[iii.]
Proof that $\mathcal{S}_{4,6,1,5}\geq \mathcal{S}_{4,5,1,6}$. In terms of the $\delta$-distances, this inequality takes the form [cf. Eq.~(\ref{eq:Sijkldelta})]
\begin{equation}
-\sum_{p=2}^{3}\delta_p-2\delta_1+\sqrt{\left(\sum_{p=2}^{3}\delta_p\right)^2+\left(\sum_{r=3}^{6}\delta_r\right)^2}\geq -\delta_3-2\sum_{q=1}^{2}\delta_q+\sqrt{\delta_3^2+\left(\sum_{r=2}^{6}\delta_r\right)^2},
\end{equation}
which, letting $\Delta_{456}\coloneq \delta_4+\delta_5+\delta_6$, can be trivially rearranged into
\begin{equation}\label{eq:itemiiiorinosq}
\delta_2+\sqrt{(\delta_2+\delta_3)^2+(\delta_3+\Delta_{456})^2}\geq\sqrt{\delta_3^2+(\delta_2+\delta_3+\Delta_{456})^2}.
\end{equation}
The validity of Eq.~(\ref{eq:itemiiiorinosq}) is certified by Eq.~(\ref{eq:ineqabc}) (Lemma~\ref{lemma1}) with $a=\delta_2$, $b=\delta_{3}$, and $c=\Delta_{456}$.

\item[iv.]
Proof that $\mathcal{S}_{4,6,1,5}\geq \mathcal{S}_{5,6,1,4}$. In terms of the $\delta$-distances, this inequality takes the form [cf. Eq.~(\ref{eq:Sijkldelta})]

\begin{equation}
-\sum_{p=2}^{3}\delta_p-2\delta_1+\sqrt{\left(\sum_{p=2}^{3}\delta_p\right)^2+\left(\sum_{r=3}^{6}\delta_r\right)^2}\geq -\delta_2-2\delta_1+\sqrt{\delta_2^2+\left(\sum_{r=4}^{6}\delta_r\right)^2},
\end{equation}
which, letting $\Delta_{456}\coloneq \delta_4+\delta_5+\delta_6$, can be trivially rearranged as
\begin{equation}\label{eq:ineqorinosq}
  \sqrt{(\delta_2+\delta_3)^2+\left(\Delta_{456}+\delta_3\right)^2}\geq \sqrt{\delta_2^2+\Delta_{456}^2}+ \delta_3 .
\end{equation}
The validity of Eq.~(\ref{eq:ineqorinosq}) is certified by Eq.~(\ref{eq:ineqabc2}) (Lemma~\ref{lemma1}) with $a=\delta_3$, $b=\delta_{2}$, and $c=\Delta_{456}$.

\item[v.] Proof that $\mathcal{S}_{4,6,1,5}\geq \mathcal{S}_{1,j,k,6}$ for ($j=2$ and $k=3$) or ($j\in\{3,4,5\}$ and $k=2$). In terms of the $\delta$-distances, this inequality takes the form [cf. Eq.~(\ref{eq:Sijkldelta})]
\begin{equation}
-\sum_{p=2}^{3}\delta_p-2\delta_1+\sqrt{\left(\sum_{p=2}^{3}\delta_p\right)^2+\left(\sum_{r=3}^{6}\delta_r\right)^2}\geq
-\sum_{p=7-(j-1)}^{6}\delta_p-2\sum_{q=1}^{7-j}\delta_q+\sqrt{\left(\sum_{p=7-(j-1)}^{6}\delta_p\right)^2+\left(\sum_{r=2}^{7-k}\delta_r\right)^2},
\end{equation}
which, letting $\Delta_{23}\coloneq\delta_2+\delta_3$, $\Delta_{456}\coloneq \delta_4+\delta_5+\delta_6$, and $\Delta_{3456}\coloneq \delta_3+\delta_4+\delta_5+\delta_6$, can be trivially rearranged into
\begin{equation}\label{eq:ineqcaseviori}
\Delta_{456}+\sum_{p=2}^{7-j}\delta_p+\sqrt{\Delta_{23}^2+\Delta_{3456}^2}
\geq
\sqrt{\left(\sum_{p=7-(j-1)}^6\delta_p\right)^2+\left(\sum_{r=2}^{7-k}\delta_r\right)^2}.
\end{equation}
Now, it suffices to prove this inequality for the allowed values of $j$ and $k$ that minimize the lhs and maximize the rhs, namely $j=5$ and $k=2$, which gives
\begin{equation}
\Delta_{456}+\delta_2+\sqrt{\Delta_{23}^2+\Delta_{3456}^2}
\geq
\sqrt{\Delta_{3456}^2+\Delta_{2345}^2}\,.
\end{equation}
To see that this holds, observe that
\begin{equation}
\Delta_{456}+\delta_2+\sqrt{\Delta_{23}^2+\Delta_{3456}^2}\geq \Delta_{456}+\sqrt{\Delta_{23}^2+\Delta_{3456}^2}\geq \sqrt{(\Delta_{23}+\Delta_{456})^2+\Delta_{3456}^2}\geq \sqrt{\Delta_{2345}^2+\Delta_{3456}^2}
\end{equation}
where the first inequality follows from the fact that $\delta_2\geq 0$, the last inequality follows from the fact that $\delta_6\geq 0$ and the intermediary inequality follows from Eq.~(\ref{eq:ineqabc3}) (Lemma~\ref{lemma1}) with $a=\Delta_{456}$, $b=\Delta_{23}$, and $c=\Delta_{3456}$.

\item[vi.] Proof that $\mathcal{S}_{4,6,1,5}\geq \mathcal{S}_{1,6,2,5}$. In terms of the $\delta$-distances, this inequality takes the form [cf. Eq.~(\ref{eq:Sijkldelta})]
\begin{equation}
-\sum_{p=2}^{3}\delta_p-2\delta_1+\sqrt{\left(\sum_{p=2}^{3}\delta_p\right)^2+\left(\sum_{r=3}^{6}\delta_r\right)^2}\geq
-\sum_{p=2}^{6}\delta_p-2\delta_1+\sqrt{\left(\sum_{p=2}^{6}\delta_p\right)^2+\left(\sum_{r=3}^{5}\delta_r\right)^2},
\end{equation}
which, letting $\Delta_{23}\coloneq\delta_2+\delta_3$, $\Delta_{345}\coloneq \delta_3+\delta_4+\delta_5$, and $\Delta_{456}\coloneq \delta_4+\delta_5+\delta_6$, can be trivially rearranged into
\begin{equation}\label{eq:ineqcaseviori2}
\Delta_{456}+\sqrt{\Delta_{23}^2+\left(\Delta_{345}+\delta_6\right)^2}
\geq
\sqrt{\left(\Delta_{23}+\Delta_{456}\right)^2+\Delta_{345}^2}.
\end{equation}
To see that this inequality holds, note that
\begin{equation}
\Delta_{456}+\sqrt{\Delta_{23}^2+\left(\Delta_{345}+\delta_6\right)^2}
\geq
\Delta_{456}+\sqrt{\Delta_{23}^2+\Delta_{345}^2}
\geq
\sqrt{\left(\Delta_{23}+\Delta_{456}\right)^2+\Delta_{345}^2},
\end{equation}
where the first inequality follows from the fact that $\delta_6\geq 0$ and the second is certified by Eq.~(\ref{eq:ineqabc3}) (Lemma~\ref{lemma1}) with $a=\Delta_{456}$, $b=\Delta_{23}$, and $c=\Delta_{345}$.

\item[vii.] Proof that $\mathcal{S}_{4,6,1,5}\geq \mathcal{S}_{i,6,1,5}$ for $i\in\{2,3\}$. In terms of the $\delta$-distances, this inequality takes the form [cf. Eq.~(\ref{eq:Sijkldelta})]
\begin{equation}
-\sum_{p=2}^{3}\delta_p-2\delta_1+\sqrt{\left(\sum_{p=2}^{3}\delta_p\right)^2+\left(\sum_{r=3}^{6}\delta_r\right)^2}\geq
-\sum_{p=2}^{7-i}\delta_p-2\delta_1+\sqrt{\left(\sum_{p=2}^{7-i}\delta_p\right)^2+\left(\sum_{r=3}^{6}\delta_r\right)^2},
\end{equation}
which, letting $\Delta_{23}\coloneq \delta_2+\delta_3$ and $\Delta_{3456}\coloneq \delta_3+\delta_4+\delta_5+\delta_6$, can be trivially rearranged into
\begin{equation}\label{eq:ineqcaseviiori}
\sum_{p=4}^{7-i}\delta_p+\sqrt{\Delta_{23}^2+\Delta_{3456}^2}
\geq
\sqrt{\left(\sum_{p=2}^{7-i}\delta_p\right)^2+\Delta_{3456}^2}.
\end{equation}
The validity of Eq.~(\ref{eq:ineqcaseviiori}) is certified by Eq.~(\ref{eq:ineqabc3}) (Lemma~\ref{lemma1}) with $a=\sum_{p=4}^{7-i}\delta_p$, $b=\Delta_{23}$, and $c=\Delta_{3456}$.
\end{itemize}
\twocolumngrid
\section{Optimality proofs for rank-deficient and degenerate X-MEMS wrt purity}\label{app:XMEMSwrtP}

\subsection{Proof of Theorem~\ref{thm:rank2}}

We start by applying the extra constraints $\lambda_3=\lambda_4=\lambda_5=\lambda_6=0$ to optimization problem (\ref{eq:optgeneral}), which yields the quadratically constrained linear program (QCLP),
\begin{align}
\mbox{maximize}&\quad \lambda_1\nonumber\\
\mbox{such that}&\quad \lambda_{2}\geq 0\,,\quad\lambda_1+\lambda_2=1\,,\quad \lambda_1^2+\lambda_2^2\leq P\,.\label{eq:optrank2_1}
\end{align}
Since $\lambda_6$ is no longer a variable (we have set $\lambda_6=0$), we have replaced the constraint $\lambda_6\geq 0$ from problem~(\ref{eq:optgeneral}) with the constraint of nonnegativity of the smallest variable of the current problem; $\lambda_2\geq 0$.

Solving the equality constraint for $\lambda_2$ and applying the solution $\lambda_2=1-\lambda_1$ to the remaining constraints gives the single-variable QCLP,
\begin{align}
\mbox{maximize}&\quad \lambda_1\nonumber\\
\mbox{such that}&\quad 1-\lambda_{1}\geq 0\,,\quad P-1+2\lambda_1-2\lambda_1^2\geq 0\,.\label{eq:optrank2_2}
\end{align}
Being the quadratic constraint of the (Schur complement) form $A_{22}-A_{12}^\dagger A_{11}^{-1} A_{12}\geq 0$, with $A_{22}\coloneq P-1+2\lambda_1$, $A_{12}\coloneq\lambda_1$ and $A_{11}^{-1}\coloneq 2$, Eq.~(\ref{eq:optrank2_2}) can be rewritten as the linear matrix inequality $A\geq 0$~\cite[pp. 650,651]{04Boyd} ($A$ is the block matrix made up of $A_{11}$, $A_{12}$, $A_{21}=A_{12}^\dagger$ and $A_{22}$), in which case problem~(\ref{eq:optrank2_2}) becomes
\begin{align}
\mbox{maximize}&\quad \lambda_1\nonumber\\
\mbox{such that}&\quad \left[\begin{array}{c|cc}
1-\lambda_1&\cdot&\cdot\\\hline
\cdot&\frac{1}{2}&\lambda_1\\
\cdot&\lambda_1&P-1+2\lambda_1
\end{array}\right]\geq 0\,.\label{eq:optrnk2_step3}
\end{align}

Now, problem~(\ref{eq:optrnk2_step3}) is an SDP and is taken to be our primal problem. In the SDP inequality form, it is
\begin{equation}\label{eq:optrnk2_step4}
-\min_{\lambda_1}\quad -\lambda_1\quad\mbox{such that} \quad \bm{F}_0 + \lambda_1\bm{F}_1 \geq 0\,,
\end{equation}
where we have defined
\begin{equation}
\bm{F}_0\coloneq \left[\begin{array}{c|cc}
1&\cdot&\cdot\\\hline
\cdot&\frac{1}{2}&\cdot\\
\cdot&\cdot&P-1
\end{array}\right]\quad\mbox{and}\quad
\bm{F}_1\coloneq \left[\begin{array}{c|cc}
-1&\cdot&\cdot\\\hline
\cdot&\cdot&1\\
\cdot&1&2
\end{array}\right]\,.
\end{equation}
It is simple to check that $\widetilde{\lambda}_1$ [cf. Eq.~(\ref{eq:lambda1lambda2rank2})] is a primal feasible point\footnote{The spectrum of $\bm{F}_0 + \widetilde{\lambda}_1\bm{F}_1$ is $\{0, (1-f_P)/2, 1+f_P+f_P^2/2\}$, which is nonnegative for $P~\in~[1/2,1[$.} yielding the primal feasible value $N_{{\rm X},P}^{(2)}$ [cf. Eq.~(\ref{eq:NXP2})]. Next, relying on duality arguments, we show that $\widetilde{\lambda}_1$ is actually an optimal primal solution.

The associated dual problem (in the matrix variable $\bm{Z}$) written in the SDP standard form is 
\begin{equation}
-\max_{\bm{Z}} -\tr[\bm{F}_0\bm{Z}]
\quad\mbox{such that } \left\{
\begin{array}{rcl}
\bm{Z}&\geq& 0\,,\\
\tr[\bm{F}_1 \bm{Z}]&=&-1\,.
\end{array}\right.
\end{equation}
According to the weak-duality property for SDPs, the optimality of $\widetilde{\lambda}_1$ is certified by the existence of a matrix $\widetilde{\bm{Z}}$ that complies with the dual problem constraints and satisfies $\tr[\bm{F}_0\widetilde{\bm{Z}}]=N_{{\rm X},P}^{(2)}$.

For $P=1/2$, consider
\begin{equation}\label{eq:Zpmeio}
\widetilde{\bm{Z}}= \left[\begin{array}{c|cc}
\cdot&\cdot&\cdot\\\hline
\cdot&\mathfrak{z}&\frac{1}{2}-\mathfrak{z}\\
\cdot&\frac{1}{2}-\mathfrak{z}&\mathfrak{z}-1
\end{array}\right]\!,
\end{equation}
in which case, regardless of the value of $\mathfrak{z}\in\mathbb{R}$, it follows that $\tr[\bm{F}_1\widetilde{\bm{Z}}]=-1$ and $\tr[\bm{F}_0\widetilde{\bm{Z}}]=N_{{\rm X},\frac{1}{2}}^{(2)}=\frac{1}{2}$. Moreover, the nonzero eigenvalues of matrix (\ref{eq:Zpmeio}) are
\begin{equation}
\Lambda_\pm(\mathfrak{z})=\mathfrak{z}-\frac{1}{2}\pm\sqrt{\mathfrak{z}^2-\left(\mathfrak{z}-\frac{1}{2}\right)}\,,
\end{equation}
which can be easily shown to satisfy
\begin{equation}
\Lambda_+(\mathfrak{z})\geq 0 \quad\mbox{and}\quad \lim_{\mathfrak{z}\to\infty}\Lambda_-(\mathfrak{z})=0\,.
\end{equation}
Therefore, the matrix inequality $\widetilde{\bm{Z}}\geq 0$ holds asymptotically (as $\mathfrak{z}\to\infty$).

For $P~\in~]\frac{1}{2},1[$, consider
\begin{equation}
\widetilde{\bm{Z}}= \left[\begin{array}{c|cc}
\cdot&\cdot&\cdot\\\hline
\cdot&1+\frac{f_P}{2}+\frac{1}{2f_P}&\T -\frac{1}{2}\left(1+\frac{1}{f_P}\right)\\
\cdot&\T-\frac{1}{2}\left(1+\frac{1}{f_P}\right)&\frac{1}{2f_P}
\end{array}\right]\!.
\end{equation}
The only nonzero eigenvalue of $\widetilde{\bm{Z}}$ is $1+f_P/2+1/f_P$, which is a strictly positive real function for $P>1/2$. Besides, it is straightforward to check that also in this case $\tr[\bm{F}_1\widetilde{\bm{Z}}]=-1$ and $\tr[\bm{F}_0\widetilde{\bm{Z}}]=N_{{\rm X},P}^{(2)}$. Thus, $\widetilde{\lambda}_1$ is a primal optimal solution.

We conclude by noting that $\widetilde{\lambda}_2$ given in Eq.~(\ref{eq:lambda1lambda2rank2}) is simply $1-\widetilde{\lambda}_1$ (normalization) and that Eq.~(\ref{eq:rhoXP2}) is simply Eq.~(\ref{eq:egXMEMSwrtspec}) with $\lambda_3$, $\lambda_4$, $\lambda_5$, and $\lambda_6$ set to zero and $\lambda_1$ and $\lambda_2$ set to $\widetilde{\lambda}_1$ and $\widetilde{\lambda}_2$, respectively.

\subsection{Proof of Theorem~\ref{thm:rank3}}

Applying the extra constraints $\lambda_4=\lambda_5=\lambda_6=0$ to optimization problem (\ref{eq:optgeneral}) and following analogous steps as those taken in the proof of Theorem~\ref{thm:rank2} (i.e., replacing $\lambda_6\geq 0$ with $\lambda_3\geq 0$, using the normalization condition $\lambda_3=1-\lambda_1-\lambda_2$ to eliminate the variable $\lambda_3$, and recognizing a Schur complement to rewrite the remaining constraints as an LMI), problem (\ref{eq:optgeneral}) takes the form of an SDP (primal problem) which, in the inequality form, is given by 
\begin{equation}\label{eq:optrnk3_step4}
-\min_{\lambda_1,\lambda_2}\quad\!\! -\lambda_1\quad\mbox{such that} \quad \bm{F}_0 + \lambda_1\bm{F}_1+\lambda_2\bm{F}_2 \geq 0,
\end{equation}
where 
\begin{equation}
\bm{F}_0\coloneq \left[\begin{array}{c|ccc}
1&\cdot&\cdot&\cdot\\\hline
\cdot&\T\frac{2}{3}&-\frac{1}{3}&\cdot\\
\cdot&\T-\frac{1}{3}&\frac{2}{3}&\cdot\\
\cdot&\cdot&\cdot&P-1
\end{array}\right]\!,\nonumber
\end{equation}
\begin{equation}
\bm{F}_1\coloneq \left[\begin{array}{c|ccc}
-1&\cdot&\cdot&\cdot\\\hline
\cdot&\cdot&\cdot&1\\
\cdot&\cdot&\cdot&\cdot\\
\cdot&1&\cdot&2
\end{array}\right]\!,\quad\mbox{and}\quad
\bm{F}_2\coloneq \left[\begin{array}{c|ccc}
-1&\cdot&\cdot&\cdot\\\hline
\cdot&\cdot&\cdot&\cdot\\
\cdot&\cdot&\cdot&1\\
\cdot&\cdot&1&2
\end{array}\right]\!.
\end{equation}
A straightforward computation of the eigenvalues of $\bm{F}_0+\widetilde{\lambda}_1\bm{F}_1+\widetilde{\lambda}_2\bm{F}_2$\footnote{The spectrum of $\bm{F}_0+\widetilde{\lambda}_1\bm{F}_1+\widetilde{\lambda}_2\bm{F}_2$ is $\{0, (2-g_P)/6,\alpha_-,\alpha_+\}$, where $$\alpha_{\pm}\coloneq\frac{1}{12}\left(12+2g_P+g_P^2\pm g_P\sqrt{16+4g_P+g_P^2}\right)\,.$$ All of these eigenvalues can be easily shown to be nonnegative for $P\in[1/3,1[$.} shows that ($\widetilde{\lambda}_1$,$\widetilde{\lambda}_2$) [cf. Eqs.~(\ref{eq:lambda1lambda2rank3})] is a primal feasible point with a primal feasible value $N_{{\rm X},P}^{(3)}$ [cf. Eq.~(\ref{eq:NXP3})]. To see that this is actually an optimal primal point, consider the associated dual problem,
\begin{equation}
-\max_{\bm{Z}} -\tr[\bm{F}_0\bm{Z}]
\quad\mbox{such that } \left\{
\begin{array}{rcl}
\bm{Z}&\geq& 0\,,\\
\tr[\bm{F}_1 \bm{Z}]&=&-1\,,\\
\tr[\bm{F}_2 \bm{Z}]&=&0\,.
\end{array}\right.
\end{equation}
In this framework, the optimality of ($\widetilde{\lambda}_1$,$\widetilde{\lambda}_2$) is certified by the existence of a matrix $\widetilde{\bm{Z}}$ that satisfies the dual constraints and satisfies $\tr[\bm{F}_0\widetilde{\bm{Z}}]=N_{{\rm X},P}^{(3)}$.

For $P=1/3$, consider
\begin{equation}\label{eq:Zpmeio2}
\widetilde{\bm{Z}}\coloneq \left[\begin{array}{c|ccc}
\cdot&\cdot&\cdot&\cdot\\\hline
\cdot&\mathfrak{z}&\T\mathfrak{z}-\frac{1}{2}&-\mathfrak{z}+\frac{1}{2}\\
\cdot&\mathfrak{z}-\frac{1}{2}&\mathfrak{z}-1&-\mathfrak{z}+1\\
\cdot&-\mathfrak{z}+\frac{1}{2}&-\mathfrak{z}+1&\mathfrak{z}-1
\end{array}\right]\!,
\end{equation}
which, regardless of the value of $\mathfrak{z}\in\mathbb{R}$, satisfies $\tr[\bm{F}_1\widetilde{\bm{Z}}]=-1$, $\tr[\bm{F}_2\widetilde{\bm{Z}}]=0$, and $\tr[\bm{F}_0\widetilde{\bm{Z}}]=N_{{\rm X},\frac{1}{3}}^{(3)}=\frac{1}{3}$. Moreover, the nonzero eigenvalues of matrix (\ref{eq:Zpmeio2}) are
\begin{equation}
\Lambda_\pm(\mathfrak{z})=\frac{3\mathfrak{z}}{2}-1\pm\frac{\sqrt{3}}{2}\sqrt{3\mathfrak{z}^2-4\mathfrak{z}+2}\,,
\end{equation}
which can be easily shown to satisfy
\begin{equation}
\Lambda_+(\mathfrak{z})\geq 0 \quad\mbox{and}\quad \lim_{\mathfrak{z}\to\infty}\Lambda_-(\mathfrak{z})=0\,.
\end{equation}
Therefore, the matrix inequality $\widetilde{\bm{Z}}\geq 0$ holds asymptotically (as $\mathfrak{z}\to\infty$).

For $P\in\,]1/3,1[$, consider
\begin{equation}
\widetilde{\bm{Z}}\coloneq \left[\begin{array}{c|ccc}
\cdot&\cdot&\cdot&\cdot\\\hline
\cdot&\T{\scriptstyle 1}+\frac{g_P}{4}+\frac{1}{g_P}&\frac{1}{2}+\frac{1}{g_P}&-\frac{1}{2}-\frac{1}{g_P}\\
\cdot&\frac{1}{2}+\frac{1}{g_P}&\frac{1}{g_P}&-\frac{1}{g_P}\\
\cdot&-\frac{1}{2}-\frac{1}{g_P}&-\frac{1}{g_P}&\frac{1}{g_P}\\
\end{array}\right]\!.
\end{equation}
The only nonzero eigenvalue of $\widetilde{\bm{Z}}$ is $1+\frac{g_P}{4}+\frac{3}{g_P}$, which is a strictly positive real function for $P>1/3$. Besides, it is straightforward to check that also in this case $\tr[\bm{F}_1\widetilde{\bm{Z}}]=-1$, $\tr[\bm{F}_2\widetilde{\bm{Z}}]=0$ and $\tr[\bm{F}_0\widetilde{\bm{Z}}]=N_{{\rm X},P}^{(3)}$. Thus, ($\widetilde{\lambda}_1$,$\widetilde{\lambda}_2$) indeed constitutes a primal optimal point.

We conclude by noting that $\widetilde{\lambda}_3$ given in Eq.~(\ref{eq:lambda1lambda2rank3}) is simply $1-\widetilde{\lambda}_1-\widetilde{\lambda}_2$ (normalization) and that Eq.~(\ref{eq:rhoXP3}) is simply Eq.~(\ref{eq:egXMEMSwrtspec}) with $\lambda_4$, $\lambda_5$ and $\lambda_6$ set to zero and $\lambda_1$, $\lambda_2$ and $\lambda_3$ set to $\widetilde{\lambda}_1$, $\widetilde{\lambda}_2$, and $\widetilde{\lambda}_3$, respectively.

\onecolumngrid
\subsection{Proof of Theorem~\ref{thm:deg}}

Once again, we closely follow the steps taken in the proofs of Theorems~\ref{thm:rank2} and \ref{thm:rank3}. First, we apply the extra constraints $\lambda\coloneq \lambda_4=\lambda_5=\lambda_6$ to optimization problem (\ref{eq:optgeneral}), which linearizes the objective function to $\lambda_1-3\lambda$. Then, we replace the constraint $\lambda_6\geq 0$ with $\lambda\geq 0$ and, afterwards, we eliminate the variable $\lambda$ by means of the normalization condition $3\lambda=1-\lambda_1-\lambda_2-\lambda_3$. Finally, recognizing a Schur complement and rewriting the inequality constraints as an LMI, we are left with the following SDP (primal problem):
\begin{equation}
-1-\min_{\lambda_1,\lambda_2,\lambda_3} -2\lambda_1-\lambda_2-\lambda_3\quad\mbox{such that } \bm{F}_0+\lambda_1\bm{F}_1+\lambda_2\bm{F}_2+\lambda_3\bm{F}_3\geq 0,\label{eq:optdeg_step4}
\end{equation}

with
\begin{equation}
\begin{array}{cccc}
\bm{F}_0\coloneq \left[\begin{array}{c|cccc}
1&\cdot&\cdot&\cdot&\cdot\\\hline
\cdot&\T\frac{5}{6}&-\frac{1}{6}&-\frac{1}{6}&\cdot\\
\cdot&\T-\frac{1}{6}&\frac{5}{6}&-\frac{1}{6}&\cdot\\
\cdot&\T-\frac{1}{6}&-\frac{1}{6}&\frac{5}{6}&\cdot\\
\cdot&\cdot&\cdot&\cdot&P-\frac{1}{3}
\end{array}\right]\!,&
\bm{F}_1\coloneq \left[\begin{array}{c|cccc}
-1&\cdot&\cdot&\cdot&\cdot\\\hline
\cdot&\cdot&\cdot&\cdot&1\\
\cdot&\cdot&\cdot&\cdot&\cdot\\
\cdot&\cdot&\cdot&\cdot&\cdot\\
\cdot&1&\cdot&\cdot&\frac{2}{3}
\end{array}\right]\!,& 
\bm{F}_2\coloneq \left[\begin{array}{c|cccc}
-1&\cdot&\cdot&\cdot&\cdot\\\hline
\cdot&\cdot&\cdot&\cdot&\cdot\\
\cdot&\cdot&\cdot&\cdot&1\\
\cdot&\cdot&\cdot&\cdot&\cdot\\
\cdot&\cdot&1&\cdot&\frac{2}{3}
\end{array}\right]\!,&\mbox{and}\quad
\bm{F}_3\coloneq \left[\begin{array}{c|cccc}
-1&\cdot&\cdot&\cdot&\cdot\\\hline
\cdot&\cdot&\cdot&\cdot&\cdot\\
\cdot&\cdot&\cdot&\cdot&\cdot\\
\cdot&\cdot&\cdot&\cdot&1\\
\cdot&\cdot&\cdot&1&\frac{2}{3}
\end{array}\right]\!.
\end{array}
\end{equation}
Computation of the eigenvalues of $\bm{F}_0+\widetilde{\lambda}_1\bm{F}_1+\widetilde{\lambda}_2\bm{F}_2+\widetilde{\lambda}_3\bm{F}_3$ can be analytically carried out\footnote{For $P\in]1/5,3/8[$, the spectrum of $\bm{F}_0+\widetilde{\lambda}_1\bm{F}_1+\widetilde{\lambda}_2\bm{F}_2+\widetilde{\lambda}_3\bm{F}_3$ is $\{1, 0, \frac{1}{2}-h_P, \beta_-, \beta_+\}$, where 
\begin{displaymath}
\beta_\pm=\frac{1}{12}\left[10+4h_P+5h_P^2\pm\sqrt{4-16h_P+8h_P^2+40h_P^3+25h_P^4}\right]\nonumber\,.
\end{displaymath}
For $P\in[3/8,1[$, the spectrum is $\{1,0,0,\gamma_-,\gamma_+\}$, where
\begin{displaymath}
\gamma_\pm=\frac{1}{12}\left[13+g_P^2\pm\sqrt{1+14 g_P^2+g_P^4}\right]\,.
\end{displaymath}
In both cases, it is simple to show that the eigenvalues are always nonnegative in the relevant purity domains.
} to confirm $(\widetilde{\lambda}_1,\widetilde{\lambda}_2,\widetilde{\lambda}_3)$ [cf. Eqs.~(\ref{eq:lambdasdeg})] as a primal feasible point. Moreover, it is simple to show that the associated primal feasible value ($-1+2\widetilde{\lambda}_1+\widetilde{\lambda}_2+\widetilde{\lambda}_3$) matches $N_{{\rm X},P}^{(\text{deg})}$ [cf. Eq.~(\ref{eq:NXPdeg})]. 

As before, we now employ SDP weak duality to show that $(\widetilde{\lambda}_1,\widetilde{\lambda}_2,\widetilde{\lambda}_3)$ is an \emph{optimal} feasible point. The associated dual problem to primal problem~(\ref{eq:optdeg_step4}) is
\begin{equation}
-1-\max_{\bm{Z}} -\tr[\bm{F}_0\bm{Z}]
\quad\mbox{such that } \left\{
\begin{array}{rcl}
\bm{Z}&\geq& 0\,,\\
\tr[\bm{F}_1 \bm{Z}]&=&-2\,,\\
\tr[\bm{F}_2 \bm{Z}]&=&-1\,,\\
\tr[\bm{F}_3 \bm{Z}]&=&-1\,,
\end{array}\right.,
\end{equation}
and any $\widetilde{\bm{Z}}$ satisfying the dual constraints and $\tr[\bm{F}_0\widetilde{\bm{Z}}]=N_{{\rm X},P}^{(\text{deg})}$ will ensure the optimality of $(\widetilde{\lambda}_1,\widetilde{\lambda}_2,\widetilde{\lambda}_3)$.

For $P\in\,]1/5,3/8[$, consider
\begin{equation}
\widetilde{\bm{Z}}\coloneq \left[\begin{array}{c|cccc}
\cdot&\cdot&\cdot&\cdot&\cdot\\\hline
\cdot&\frac{2}{3} +{\scriptstyle h_P}+\frac{1}{9h_P}&\T\frac{1}{2}+\frac{h_P}{2}+\frac{1}{9h_P}&\frac{1}{2}+\frac{h_P}{2}+\frac{1}{9h_P}&-{\scriptstyle 1}-\frac{1}{3 h_P}\\
\cdot&\frac{1}{2}+\frac{h_P}{2}+\frac{1}{9h_P}&\frac{1}{3}+\frac{h_P}{4}+\frac{1}{9h_P}&\frac{1}{3}+\frac{h_P}{4}+\frac{1}{9h_P}&-\frac{1}{2}-\frac{1}{3 h_P}\\
\cdot&\frac{1}{2}+\frac{h_P}{2}+\frac{1}{9h_P}&\frac{1}{3}+\frac{h_P}{4}+\frac{1}{9h_P}&\frac{1}{3}+\frac{h_P}{4}+\frac{1}{9h_P}&-\frac{1}{2}-\frac{1}{3 h_P}\\
\cdot&-{\scriptstyle 1}-\frac{1}{3 h_P}&-\frac{1}{2}-\frac{1}{3 h_P}&-\frac{1}{2}-\frac{1}{3 h_P}&\frac{1}{h_P}
\end{array}\right]\!.
\end{equation}
The only nonzero eigenvalue of $\widetilde{\bm{Z}}$ is $\frac{4}{3}+\frac{3h_P}{2}+\frac{4}{3 h_P}$, which is a strictly positive real function for $P>1/6$. In addition, it is elementary to show that $\tr[\bm{F}_1\widetilde{\bm{Z}}]=2\tr[\bm{F}_2\widetilde{\bm{Z}}]=2\tr[\bm{F}_3\widetilde{\bm{Z}}]=-2$ and $\tr[\bm{F}_0\widetilde{\bm{Z}}]=N_{{\rm X},P}^{(\text{deg})}$, as required. 

For $P\in\,[3/8,1[$, consider
\begin{equation}
\widetilde{\bm{Z}}\coloneq \left[\begin{array}{c|cccc}
\frac{4}{3}-\frac{2}{3g_P}&\cdot&\cdot&\cdot&\cdot\\\hline
\cdot&\T\frac{4}{9}+\frac{g_P}{9}+\frac{4}{9g_P}&\frac{1}{9}-\frac{g_P}{18}+\frac{4}{9g_P}&\frac{1}{9}-\frac{g_P}{18}+\frac{4}{9g_P}&-\frac{1}{3}-\frac{2}{3 g_P}\\
\cdot&\frac{1}{9}-\frac{g_P}{18}+\frac{4}{9g_P}&-\frac{2}{9}+\frac{g_P}{36}+\frac{4}{9g_P}&-\frac{2}{9}+\frac{g_P}{36}+\frac{4}{9g_P}&\frac{1}{6}-\frac{2}{3 g_P}\\
\cdot&\frac{1}{9}-\frac{g_P}{18}+\frac{4}{9g_P}&-\frac{2}{9}+\frac{g_P}{36}+\frac{4}{9g_P}&-\frac{2}{9}+\frac{g_P}{36}+\frac{4}{9g_P}&\frac{1}{6}-\frac{2}{3 g_P}\\
\cdot&-\frac{1}{3}-\frac{2}{3 g_P}&\frac{1}{6}-\frac{2}{3 g_P}&\frac{1}{6}-\frac{2}{3 g_P}&\frac{1}{g_P}
\end{array}\right]\!.
\end{equation}
In this case, there are two nonzero eigenvalues, namely $\frac{4}{3}-\frac{2g_P}{3}$ and  $\frac{g_P}{6}+\frac{7}{3 g_P}$, both of which are strictly positive real functions for the relevant values of $P$. Furthermore, the identities $\tr[\bm{F}_1\widetilde{\bm{Z}}]=2\tr[\bm{F}_2\widetilde{\bm{Z}}]=2\tr[\bm{F}_3\widetilde{\bm{Z}}]=-2$ and $\tr[\bm{F}_0\widetilde{\bm{Z}}]=N_{{\rm X},P}^{(\text{deg})}$ can be easily seen to hold in this case as well, thus confirming the optimality of $(\widetilde{\lambda}_1,\widetilde{\lambda}_2,\widetilde{\lambda}_3)$ in the entire purity domain.

Finally, we note that $\widetilde{\lambda}$ given in Eq.~(\ref{eq:lambdasdeg}) is simply $(1-\widetilde{\lambda}_1-\widetilde{\lambda}_2-\widetilde{\lambda}_3)/3$ (normalization) and that Eq.~(\ref{eq:rhoXPdeg}) is simply Eq.~(\ref{eq:egXMEMSwrtspec}) with $\lambda_4$, $\lambda_5$, and $\lambda_6$ set to $\widetilde{\lambda}$ and $\lambda_1$, $\lambda_2$, and $\lambda_3$ set to $\widetilde{\lambda}_1$, $\widetilde{\lambda}_2$, and $\widetilde{\lambda}_3$, respectively.
\twocolumngrid


\end{document}